\documentclass[onecolumn,10pt]{IEEEtran}

\usepackage{graphicx}
\usepackage{amsmath,amssymb}
\usepackage{cases}
\usepackage{color}
\usepackage{amsfonts}
\usepackage{indentfirst}
\usepackage{cite}
\usepackage{caption}
\usepackage{algorithm}
\usepackage{algpseudocode}
\usepackage{booktabs}
\usepackage{subfigure}
\usepackage{multirow}
\usepackage{amsthm}
\makeatletter
\newtheoremstyle{mythm}{3pt}{3pt}{}{16pt}{\bfseries}{:}{.5em}{}
\theoremstyle{mythm}
\newtheorem{theorem}{Theorem}
\newtheorem{example}{Example}
\newtheorem{definition}{Definition}
\newtheorem{remark}{Remark}

\newtheorem{corollary}{Corollary}
\newtheorem{lemma}{Lemma}
\newtheorem{construction}{Construction}
\newcommand{\tabincell}[2]{\begin{tabular}{@{}#1@{}}#2\end{tabular}}
\begin{document}
\title{Linear Coded Caching Scheme for Centralized Networks
\author{Minquan Cheng, Jie Li \IEEEmembership{Member~IEEE}, Xiaohu Tang \IEEEmembership{Member~IEEE}, Ruizhong Wei
}
\thanks{M. Cheng is with Guangxi Key Lab of Multi-source Information Mining $\&$ Security, Guangxi Normal University,
Guilin 541004, China, (e-mail: chengqinshi@hotmail.com).}
\thanks{J. Li is  with the Hubei Key Laboratory of Applied Mathematics, Faculty
of Mathematics and Statistics, Hubei University, Wuhan 430062, China (e-mail: jieli873@gmail.com)}
\thanks{X. Tang is with the Information Security and National Computing Grid Laboratory,
Southwest Jiaotong University, Chengdu, 610031, China (e-mail: xhutang@swjtu.edu.cn).}
\thanks{R. Wei is with Department of Computer Science, Lakehead University, Thunder Bay, ON, Canada, P7B 5E1,(e-mail: rwei@lakeheadu.ca).}
}
\date{}
\maketitle

\begin{abstract}
Coded caching systems   have been widely studied to reduce the data transmission during the peak traffic time.
In practice,  two important parameters of a coded caching system should be considered, i.e., the rate which is the maximum amount of the data transmission during the peak traffic time, and the subpacketization level, the number of divided packets of each file  when we implement a coded caching scheme. We prefer to design a scheme with rate and packet number as small as possible since they  reflect the transmission efficiency and complexity of the caching scheme, respectively.

In this paper, we first characterize a coded caching scheme from the viewpoint of linear algebra and show that designing
a linear coded caching scheme is equivalent to constructing three classes of matrices satisfying some rank conditions. Then based on the invariant linear subspaces and combinatorial design theory, several classes of new coded caching schemes over $\mathbb{F}_2$ are obtained by constructing these three classes of matrices. It turns out  that the rate of our new rate is the same as  the scheme construct by Yan et al. (IEEE Trans. Inf. Theory 63, 5821-5833, 2017), but the packet number is significantly reduced. A concatenating construction then is used for flexible number of users. Finally by means of these matrices, we show that the  minimum storage regenerating codes can also be used to construct coded caching schemes.
\end{abstract}

\begin{IEEEkeywords}
Linear coded caching scheme, matrices, rate, packet number
\end{IEEEkeywords}

%
\IEEEpeerreviewmaketitle

\section{Introduction}
\label{sec-intr}
Due to the fast growth of network applications, there exist  tremendous pressures on the data transmission over computer networks,
and the high temporal variability of network traffic further imposes these pressures. As a result, the communication systems are congested during peak-traffic hours and underutilized during off-peak hours. Recently caching {{systems}}, where some traffics are proactively placed into the user's memory during the off peak traffic hours, are regarded as an interesting solution to relieve the pressures and have been used in heterogeneous wireless networks \cite{ART,MN1,RAN,JCM,JWTL,MN,KNMD}.

A centralized caching system, which has been discussed widely, is the basic model of networks for
discussing the caching systems. In a centralized caching system, a server, which contains $N$ files with equal size, connects to $K$ users,
each of them has  cache size of $M$ files, by a single shared link. A caching scheme consists of two independent phases, i.e., {\em placement phase} during the off-peak traffic hours and {\em delivery phase} during the peak traffic hours. Assume that each user requests one file during the peak traffic hours. We prefer to design a scheme for the caching system such that in the delivery phase the data broadcasting can satisfy various demands of all users, and at the same time, the amount of the data, which is called {\em rate}, is as small as possible.

Maddah-Ali and Niesen in \cite{MN} proposed a coded caching approach. The proposed coded caching schemes, i.e., coded caching in placement phase and coded broadcasts in delivery phase,  allow   significant reductions in the number of bits transmitted from the server to the end users. A coded caching scheme called $F$-division scheme if each file is split into $F$ packets. If the packets of all files are directly cached in the placement phase, we call it uncoded placement. Otherwise we call it coded placement.
Through an elaborate uncoded placement and a coded delivery, Maddah-Ali and Niesen in \cite{MN} proposed the first deterministic scheme for an $F$-division $(K,M,N)$ coded caching system with $F={K\choose KM/N}$, when $\frac{KM}{N}$ is an integer.
 Such a scheme is referred to as MN scheme in this paper.




\subsection{Prior work}

In this paper, we are interested in the centralized coded caching schemes. So far, many results have been obtained based on that
model, for instances, \cite{MN, AG,CYTJ,GR,STC,TC,WLG,YCTC,YTCC,YMA,JCLC,WTP} etc. The following parameters are two of
the  main
concerns when researchers construct  caching schemes.

{\bf Transmission rate:} The first parameter is  the transmission rate under the assumption that each file in the library is large enough. In \cite{GR}, an improved lower bound of the transmission rate is derived by a combinatorial problem of optimally labeling the leaves of a directed tree. By interference elimination, the transmission rate of a coded caching scheme is obtained for the case $N\leq K$. There are many discussions on the transmission rate by adding another condition, i.e., under the uncoded placement. For example, \cite{WTP} showed that MN scheme has minimum rate by the graph theory, when $K< N$. \cite{YMA} obtained the transmission rate for the various demand patterns by modifying MN scheme in the delivery phase. \cite{JCLC} translated the transmission rate for the non-uniform file popularity and the various demand patterns into solving a optimization problem with $N2^K$ variables. It is interesting that when $K<N$ and  all of the files have the same popularity the solution of this optimization problem is exactly  the rate of MN scheme. But for the other cases, this optimization problem is NP-Hard.

{\bf Packet number:} The second parameter is the value of packet number $F$. In practice the packet number $F$ is finite. Furthermore, the complexity of a coded caching scheme increases when the parameter $F$ increasing. $F$ is also referred as subpachetization level by many authors.
So far, all the known discussions for the packet number are proposed under the assumption of identical uncoded caching policy, i.e.,  each user caches packets with the same indices from all files, where packets belonging to every file is ordered according to a chosen numbering (note that we assume every
file has the same size). In this paper we only consider the schemes under this assumption when $K<N$. All the following introduced previously known studies have this assumption.
 In \cite{CYTJ} authors showed that the minimum packet number is $F={K\choose KM/N}$, i.e., the packet number of MN scheme, for the minimum rate.
  It is easy to see that $F$ in MN scheme increases very quickly with the number of users $K$. This would become infeasible when $K$ becomes large. It is well known that there is tradeoff between $R$ and $F$.

The first scheme with  lower packet number was proposed by Shanmugam et al. in \cite{SJTLD}.
 Recently, Yan et al. in \cite{YCTC} characterized an $F$-division $(K,M,N)$ caching scheme by a simple array which is called $(K,F,Z,S)$ placement delivery array (PDA), where $\frac{M}{N}=\frac{Z}{F}$ and $R=\frac{S}{F}$. Then they 
generated   two infinite classes of PDAs which gives two classes of schemes.
Some other deterministic coded caching schemes with lower  subpacketization level were proposed while the  rate $R$ increased. For example, \cite{SZG} obtained two classes of schemes by constructing the special $(6,3)$-free hypergraphs. \cite{CJYT} generalized all
the constructions in \cite{YCTC} and most results in \cite{SZG} by means of PDAs.
\cite{TR} used resolvable combinatorial designs and linear block codes to construct two classes of schemes.
\cite{STD,SDLT} obtained some schemes by the known results on the $(r,t)$ Ruzsa-Szem\'{e}redi graphs.
\cite{YTCC} obtained a class of schemes by the results of strong edge coloring of bipartite graphs.
\cite{K} got a class of schemes by means of projective space over $\mathbb{F}_p$ where $p$ is an prime power.
However, up to now there are a few results on the above related graphs for some special parameters, and the existences of these graphs are open problems, especially  for the deterministic constructions. Furthermore, the authors in \cite{SDLT} pointed that all the deterministic coded caching schemes can be recast into the PDA. By means of PDA, the authors in \cite{CYTJ} obtained two classes of Pareto-optimal PDAs based on MN scheme.



As we have mentioned,  $F$ in the MN scheme has to be at least ${K\choose KM/N}$. When $K$ is large, $F$  will be very large so that MN scheme would become unfeasible in practice. Therefore researchers investigated schemes with smaller $F$ thereafter.
 On the other hand, the rate of MN scheme is almost optimal. In this paper, we focus on explicit coded caching schemes with small
 rate and reduced packet number when $K\leq N$.
To compare with the known results, we select one scheme from \cite{YCTC} and the MN scheme from \cite{MN}  in  Table \ref{tab-main}.
As we have mentioned that there are many results about the constructions of coded caching schemes. However,   up to now
only the schemes from \cite{YCTC} and \cite{MN} have the rate similar to our new scheme,
which is approximating to the optimal  rate. The rate of other schemes are
not competitive as that of our new scheme.
\begin{table}[H]
  \centering
  \caption{MN schemes and schemes with rate closing  MN scheme for any positive integer $q\geq 2$\label{tab-main}}
  \normalsize{
  \begin{tabular}{|c|c|c|c|c|}
\hline
results & $M/N$   & $R$& $F$   \\ \hline
MN scheme in \cite{MN} &$\frac{1}{q}$&$ \frac{K}{q+K}(q-1)$&
 $\sim\frac{q}{\sqrt{2\pi K(q-1)}}\cdot e^{\frac{K}{q}\left(\ln  q+\left(q-1\right)\ln \frac{q}{q-1}\right)}$\\  \hline
The scheme in \cite{YCTC}&$\frac{1}{q}$&$q-1$&$e^{\left(\frac{K}{q}-1\right)\ln q} $\\  \hline
Theorem \ref{th-main-R}  &$\frac{1}{q}$& $q-1$&$e^{\frac{K}{q+1}\ln q}$\\  \hline
   \end{tabular}}
\end{table}
 From the table, we can see that the scheme in \cite{YCTC} significantly reduces $F$ while the rate
  approximates that of MN scheme when $K$ is large.
An interesting question is that for the same or similar rate of MN scheme, do there exist some classes of schemes
 in which the packet number is smaller than that of the scheme in \cite{YCTC}. In this paper, we shall give the positive answer.
 The scheme at third row is from this paper, which further significantly  reduced the value
  of $F$.
 Our new determined scheme has the rate almost as small as MN scheme, but the packet
  number is much smaller than that in \cite{YCTC}.
   It is easy to check that for the fixed $\frac{M}{N}$, the ratio of packet numbers of
   these two schemes is $q^{\frac{K}{q(q+1)}-1}$, 
    which is
    exponential with $K$. We also provided some methods to release the restriction of $K$ so that the value of $K$ is flexible.


The main contributions of this paper are as follows. Firstly, we give a generalization of PDA method which is essentially the
basic method for most  previous constructions (see \cite{SDLT}). We characterize a general coded caching system from the view point of linear algebra. Consequently a coded caching scheme can be represented by three classes of matrices, say caching matrices, coding matrices and decoding matrices, satisfying some rank conditions. It means that we reduced the problem to
finding sequences of subspace of a linear space satisfying certain properties and interrelationships. Secondly,
as one example of implementation, we use the invariant subspaces of linear algebra and combinatorial design theory to
construct three classes of matrices  over $\mathbb{F}_2$. Thus a class of determined coded caching scheme over $\mathbb{F}_2$ is obtained.
The new scheme has good rate and significant reduced packet number. Thirdly, since many previous schemes required strict user numbers,
 we use some concatenating constructions to obtain schemes with flexible $K$.
 Finally we  demonstrate  that the regenerating code with optimal repair bandwidth, which is a hot topic in the distributed storage,
  can also be used to constructed a coded caching scheme.

The rest of this paper is organized as follows. Section \ref{sec_prob} introduces system model and backgrounds for coded caching system.
 In Section \ref{se-new character}, a new characterization of coded caching scheme is proposed from the viewpoint of matrices.
In Section \ref{sec-new-construction}, a new class of coded caching schemes is proposed.
Section \ref{sec-concatenating} introduces the concatenating construction. We indicate that the regenerating code with optimal repair bandwidth can be used to construct a coded caching scheme in Section \ref{sec_MSRC}. Finally a conclusion is drawn in Section \ref{sec_conclusion}. Some of detailed proofs are
attached in Appendix.

\section{Problem formulation and preliminaries}
\label{sec_prob}
In the centralized caching system, a single server
containing $N$ files with the same length, denoted by $\mathcal{W}$ $=\{\mathbf{W}_0$, $\mathbf{W}_1$, $\ldots$, $\mathbf{W}_{N-1}\}$, connects to $K$ users, denoted by $\mathcal{K}=\{0,1,\ldots,K-1\}$, over a shared link, and each user has a cache memory of size $M$ files (see Fig. \ref{system}).
\begin{figure}[h]
\centering\includegraphics[width=0.4\textwidth]{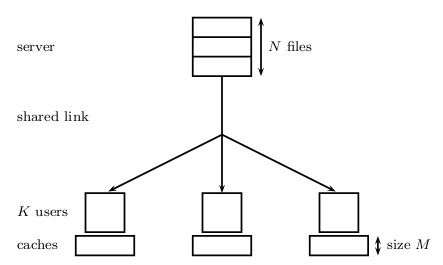}
\caption{Caching system}\label{system}
\end{figure}
An $F$-division $(K,M,N)$ coded caching scheme consists of two independent phases.
\begin{itemize}
\item Placement phase: Server divides each file into {$F$ packets, each of which has the same size,} and then places some
uncoded/coded packets of each file into users' cache memories. Denote the content cached in user $k$ by $\mathcal{Z}_k$.
In this phase server does not know the users' requests in the following phase.
\item Delivery phase: Assume that each user randomly requests one file from $N$ files independently.
Denote the request file numbers by ${\bf d}=(d_0,d_1,\cdots,d_{K-1})$, i.e., user $k$ requests the file $\mathbf{W}_{d_k}$ where $d_k\in[0,N)$ and $k\in[0,K)$. Then server sends a coded signal of size at most $S_{{\bf d}}$ packets, denoted by $\mathbf{X}_{{\bf d}}$, to
 the users such that each user's demand is satisfied with the help of the contents of its cache.
\end{itemize}

Let $R$ denote the maximum transmission amount among all the request during the delivery phase, i.e.,
\begin{align}
R=\sup_{\tiny{\mbox{$\begin{array}{c}
                 \mathbf{d}=(d_0,\cdots,d_{K-1}) \\
                 d_k\in[0,N),\forall k\in[0,K)
               \end{array}
$}}} \left\{\frac{S_{\mathbf{d}}}{F}\right\}\label{eq_def_R}.
\end{align}
$R$ is  called the rate of a coded caching scheme. 
Maddah-Ali and Niesen proposed  the first well known {deterministic} coded caching scheme.
\begin{lemma}(MN scheme, \cite{MN})
\label{le-MN}
For any positive integers $N$, $K$ and $M$ with $M\in \{N/K,2N/K,\ldots,N\}$, there exists an $F$-division $(K,M,N)$ coded caching scheme with $F={K\choose MK/N}$ and $R=\frac{K(1-M/N)}{1+KM/N}$.
\end{lemma}
Denote the minimum value of $R$ by $R^*$. Maddah-Ali and Niesen in \cite{MN} proved that the ratio between their rate
 and $R^*$ is less than or equal to a constant $12$. So MN scheme is ordered optimal. Up to now, the rate of {MN} scheme
 is the smallest rate among of all the known schemes when $K\leq N$.

\subsection{Placement delivery array}

Before introduce our new linear methods, we review the placement delivery array. Our new method is inspired from this structure.
As \cite{SDLT} indicated that  PDA method is essentially the
basic method for most  previous constructions. Our purpose is to generalize the PDA method and get better results from
the generalization.

To construct  uncoded caching, Yan et al. \cite{YCTC} generated some $F$-division $(K,M,N)$ coded caching scheme by constructing a related combinatorial structure
which is defined as follows.
\begin{definition}[\cite{YCTC}]
For positive integers $K,F, S$ and $Z$, an $F\times K$ array $\mathbf{P}=(p_{j,k})$, $j\in [0,F), k\in[0,K)$, composed of a specific symbol $``*"$  and $S$ nonnegative integers
$0,1,\cdots, S-1$, is called a $(K,F,Z,S)$ placement delivery array (PDA) if it satisfies the following conditions:
\begin{itemize}
  \item The symbol $``*"$ appears $Z$ times in each column;
  \item Each integer occurs at least once in the array;
  \item For any two distinct entries $p_{j_1,k_1}$ and $p_{j_2,k_2}$,    $p_{j_1,k_1}=p_{j_2,k_2}=s$ is an integer  only if
  \begin{itemize}
     \item $j_1\ne j_2$, $k_1\ne k_2$, i.e., they lie in distinct rows and distinct columns; and
     \item $p_{j_1,k_2}=p_{j_2,k_1}=*$, i.e., the corresponding $2\times 2$  subarray formed by rows $j_1,j_2$ and columns $k_1,k_2$ must be of the following form
  \begin{eqnarray*}
    \left(\begin{array}{cc}
      s & *\\
      * & s
    \end{array}\right)~\textrm{or}~
    \left(\begin{array}{cc}
      * & s\\
      s & *
    \end{array}\right)
  \end{eqnarray*}
   \end{itemize}
\end{itemize}
\end{definition}
\begin{example}\label{exam2} It is easily checked that the following array is a $(6,4,2,4)$ PDA.
\begin{align}
\mathbf{P}=\left(\begin{array}{cccccc}
 *&  1&  *&  2&  *&  0\\
 0&  *&  *&  3&  1&  *\\
 *&  3&  0&  *&  2&  *\\
 2&  *&  1&  *&  *&  3
\end{array}
\right)\label{eq_A}
\end{align}
\end{example}
Based on a $(K,F,Z,S)$  PDA $\mathbf{P}=(p_{j,k})$ with $j\in [0,F)$ and $k\in[0,K)$,  Yan et al., in \cite{YCTC} showed that an $F$-division caching scheme for a $(K,M,N)$ caching system with $\frac{M}{N}=\frac{Z}{F}$ can be conducted by Algorithm \ref{alg:PDA}.
\begin{algorithm}[htb]
\caption{caching scheme based on PDA in \cite{YTCC}}\label{alg:PDA}
\begin{algorithmic}[1]
\Procedure {Placement}{$\mathbf{P}$, $\mathcal{W}$}
\State Split each file $\mathbf{W}_i\in\mathcal{W}$ into $F$ packets, i.e., $\mathbf{W}_{i}=\{\mathbf{W}_{i,j}\ |\ j=0,1,\cdots,F-1\}$.
\For{$k\in \mathcal{K}$}
\State $\mathcal{Z}_k\leftarrow\{\mathbf{W}_{i,j}\ |\ p_{j,k}=*, \forall~i=0,1,\cdots,N-1\}$
\EndFor
\EndProcedure
\Procedure{Delivery}{$\mathbf{P}, \mathcal{W},{\bf d}$}
\For{$s=0,1,\cdots,S-1$}
\State  Server sends $\bigoplus_{p_{j,k}=s,0\leq j< F,0\leq k<K}\mathbf{W}_{d_{k},j}$.
\EndFor
\EndProcedure
\end{algorithmic}
\end{algorithm}
\begin{theorem}[\cite{YCTC}]\label{thm1} For a given $(K,F,Z,S)$ PDA, by Algorithm \ref{alg:PDA} there exists an $F$-division $(K,M,N)$ caching scheme with $\frac{M}{N}=\frac{Z}{F}$ and $R=\frac{S}{F}$ for any $M$ and $N$.
\end{theorem}
\begin{example}\label{exam3}
From Example \ref{exam2}, let us generate a $4$-division $(6,3,6)$ coded caching scheme by Algorithm \ref{alg:PDA}. Clearly $\mathcal{W}=\{\mathbf{W}_0,\mathbf{W}_1,\mathbf{W}_2,\mathbf{W}_3,\mathbf{W}_4,\mathbf{W}_5\}$ and each file is $\mathbf{W}_i=\{\mathbf{W}_{i,0},\mathbf{W}_{i,1},\mathbf{W}_{i,2},\mathbf{W}_{i,3}\}$, $i\in [0,6)$ by Line-2 of Algorithm \ref{alg:PDA}. According to Lines 2-5 and Lines 8-10, the two phases of a coded caching scheme can be implemented as follows:
\begin{itemize}
   \item \textbf{Placement Phase}:   The contents in each users are
       \begin{align*}
       \mathcal{Z}_0=\left\{\mathbf{W}_{i,0},\mathbf{W}_{i,2}:i\in[0,6)\right\}~~~~ \mathcal{Z}_1=\left\{\mathbf{W}_{i,1},\mathbf{W}_{i,3}:i\in[0,6)\right\}\notag\\
       \mathcal{Z}_2=\left\{\mathbf{W}_{i,0},\mathbf{W}_{i,1}:i\in[0,6)\right\}~~~~ \mathcal{Z}_3=\left\{\mathbf{W}_{i,2},\mathbf{W}_{i,3}:i\in[0,6)\right\}\notag\\
       \mathcal{Z}_4=\left\{\mathbf{W}_{i,0},\mathbf{W}_{i,3}:i\in[0,6)\right\}~~~~ \mathcal{Z}_5=\left\{\mathbf{W}_{i,1},\mathbf{W}_{i,2}:i\in[0,6)\right\}\notag
       \end{align*}
   \item \textbf{Delivery Phase}: Assume the request vector $\mathbf{d}=(0,1,\cdots,5)$. Table \ref{table1} shows the transmitting process.
\begin{table}[!htp]
\centering
\caption{Deliver steps in example \ref{exam2} }\label{table1}
\normalsize{
\begin{tabular}{cc}
\hline
Time Slot& Transmitted Signnal  \\
\hline
$0$& $\mathbf{W}_{0,1}+ \mathbf{W}_{2,2}+ \mathbf{W}_{5,0}$\\
$1$& $\mathbf{W}_{1,0}+ \mathbf{W}_{2,3}+ \mathbf{W}_{4,1}$\\
$2$& $\mathbf{W}_{0,3}+ \mathbf{W}_{3,0}+ \mathbf{W}_{4,2}$\\
$3$& $\mathbf{W}_{1,2}+ \mathbf{W}_{3,1}+ \mathbf{W}_{5,3}$\\
\hline
\end{tabular}
}
\end{table}
\end{itemize}
\end{example}


By constructing an appropriate PDA, Yan et al., proposed the following two schemes where the rate is near to  that of MN scheme
 but the packet number is far less than that of MN scheme.
\begin{lemma}(\cite{YCTC})\label{le-Yan}
For any positive integers $q$, $m$ with $q\geq2$, there exists a $q^{m}$-division $((m+1)q,M,N)$ coded caching scheme with $\frac{M}{N}=\frac{1}{q}$, $\frac{q-1}{q}$ and rate $R=q-1$, $\frac{1}{q-1}$ respectively.
\end{lemma}

PDA is used to design most caching schemes previously.
 In the PDA ,  each column represents the cache data for a user. A `*' means that this part of the data is stored in
this user's cache. An integer indicates that this part of the data is not stored. In case the user requests data not stored in its cache
(i.e., some integer, say $s$,  in the cell of PDA), then the server look at
the corresponding row (say $i$th row) of the matrix. For those columns containing $s$, the server includes the data of its $i$th cell to the
delivery data.
The property of the PDA guarantees the user can get the requested file.

\section{Linear coded caching schemes}
\label{se-new character}

The PDA method showed that we need some algorithm to do the arrangement of the data for caching and delivery. For efficiency, we propose
linear algorithms to do that. That is the main motivation of our linear coded caching schemes. Now we start to introduce our method below.

Assume that each file is a column vector with length $F$ over a certain filed, say $\mathbb{F}_p$ for some prime power $p$, and the identical caching policy for all files is carried out for each user. When $\frac{FM}{N}$ is a positive integer let us characterize the coded caching scheme from
 the viewpoint of linear algorithms, i.e., the two phases can be written as follows.
\begin{itemize}
\item  In the placement phase, for each $k=0,1,\ldots,K-1$, the $k$th user caches
$$\mathcal{Z}_k=\{\mathbf{S}_k\mathbf{W}_n\ |\ n\in [0,N)\}$$
where $\mathbf{S}_k$, which is called a {\em caching matrix}, is an $\frac{FM}{N}\times F$ matrix  over some field $\mathbb{F}_p$. Clearly each user caches the contents of size $\frac{FM}{N}N=MF$.
\item In the delivery phase, for any fixed request ${\bf d}=(d_0,d_1,\cdots,d_{K-1})$, sever broadcasts signal
\begin{eqnarray}
\label{eq-broadcast}
\mathbf{X}_{{\bf d}}=\mathbf{A}_0 \mathbf{W}_{d_0}+\mathbf{A}_1 \mathbf{W}_{d_1}+\ldots+\mathbf{A}_{K-1} \mathbf{W}_{d{_K-1}}
\end{eqnarray}where $\mathbf{A}_k$, which is called a {\em coding matrix}, is an $RF\times F$ matrix. Then user $k$ can obtain the required file $\mathbf{W}_{d_k}$ by the contents $\mathcal{Z}_k$ and the contents
$$\mathbf{W}'_{d_k}=\mathbf{S}'_k \mathbf{X}_{{\bf d}}$$
where $\mathbf{S}'_k$, which is called a {\em decoding matrix}, is an $F(1-\frac{M}{N})\times RF$ matrix.
\end{itemize}

In this paper, a coded caching scheme which can be characterized by the above three classes of matrices, i.e., caching matrices, coding matrices and decoding matrices, is called linear caching scheme.

Next, we shall prove some   algebraic properties of the three series of matrices.
\begin{theorem}
\label{th-main}
For any request ${\bf d}$ and the related signals $\mathbf{X}_{{\bf d}}$ in \eqref{eq-broadcast}, user $k$ can obtain the required $\mathbf{W}_{d_i}$ if and only if the matrices $\mathbf{S}_k$, $\mathbf{A}_k$ and $\mathbf{S}'_k$, $k\in [0,K)$, satisfy the following conditions.
\begin{eqnarray}\label{Eqn rank condition}
\mathrm{rank}\left(
  \begin{array}{c}
    \mathbf{S}_k \\
    \mathbf{S}'_{k}\mathbf{A}_{k'}
 \end{array}
\right)=\left\{\begin{array}{lll}F,&\mbox{if}\ \    k=k'
\\
\frac{FM}{N},&\mbox{otherwise}\end{array}\right.  \mbox{~~for~~}0\le k'< K.
\end{eqnarray}
\end{theorem}
\begin{proof} For any request vector ${\bf d}$ and the related signals $\mathbf{X}_{{\bf d}}$, user $k$ can get the following contents by decoding matrix $\mathbf{S}'_{k}$
$$\mathbf{S}'_{k}\mathbf{X}_{{\bf d}}=\mathbf{S}'_{i}\mathbf{A}_{0}\mathbf{W}_{d_0}+\mathbf{S}'_{k}\mathbf{A}_{1}\mathbf{W}_{d_1}+\ldots+
\mathbf{S}'_{k}\mathbf{A}_{K-1}\mathbf{W}_{d_{K-1}}.$$
This is,
\begin{eqnarray}
\label{eq01}
\mathbf{S}'_{k}\mathbf{A}_{k}\mathbf{W}_{d_k}=\mathbf{S}'_{k}\mathbf{X}_{{\bf d}}-\sum\limits_{k'\in[0,K)\setminus \{k\}}\mathbf{S}'_{k}\mathbf{A}_{k'}\mathbf{W}_{d_{k'}}.
\end{eqnarray}
Clearly user $k$ can cancel the last item on the right side of the above equality if and only if for each $k'\in [0,K)\setminus\{k\}$ one can obtain $\mathbf{S}'_k\mathbf{A}_{k'}\mathbf{W}_{d_{k'}}$ based on user $k$'s caching contents $\mathbf{S}_k\mathbf{W}_{d_{k'}}$, i.e., $\mathbf{S}'_k\mathbf{A}_{k'}$ can be linear represented by the rows of $\mathbf{S}_k$. This implies that the following equation always holds.
\begin{eqnarray*}
\mathrm{rank}\left(
  \begin{array}{c}
    \mathbf{S}_k \\
    \mathbf{S}'_{k}\mathbf{A}_{k'}
 \end{array}
\right)=\frac{FM}{N}.
\end{eqnarray*}
Furthermore, user $k$ can obtain the $F$ packets of $\mathbf{W}_{d_k}$ if and only if the following equation holds.
\[\mathrm{rank}\left(
  \begin{array}{c}
    \mathbf{S}_k \\
    \mathbf{S}'_k\mathbf{A}_{k}
  \end{array}
\right)=F.
\]
The proof is complete.
\end{proof}
\begin{remark}
By \eqref{Eqn rank condition} in Theorem \ref{th-main}, we know that the matrices $\mathbf{S}_k$ and $\mathbf{S}'_k$, $k\in [0,K)$, must be full row rank. From the knowledge of linear algebra, for each $k'\neq k$, $k'\in [0,K)$, we could get matrix $\mathbf{D}_{k,k'}$ satisfying
\begin{eqnarray}
\label{eq-decoding_matrix_1}
\mathbf{D}_{k,k'}\mathbf{S}_0=\mathbf{S}'_k\mathbf{A}_{k'}.
\end{eqnarray}
So from the above equation, \eqref{eq01} can be written as follows.
\begin{eqnarray*}
\label{eq01-1}
\mathbf{S}'_{k}\mathbf{A}_{k}\mathbf{W}_{d_k}=\mathbf{S}'_{k}\mathbf{X}_{{\bf d}}-\sum\limits_{k'\in[0,K)\setminus \{k\}}\mathbf{D}_{k,k'}\mathbf{S}_0\mathbf{W}_{d_{k'}}.
\end{eqnarray*}
Then
\begin{eqnarray*}
 \left(
  \begin{array}{c}
    \mathbf{S}_k \mathbf{W}_{d_k}\\
    \mathbf{S}'_{k}\mathbf{A}_{k}\mathbf{W}_{d_k}
 \end{array}
\right)= \left(
  \begin{array}{c}
    \mathbf{S}_k\\
    \mathbf{S}'_{k}\mathbf{A}_{k}
 \end{array}
\right) \mathbf{W}_{d_k}= \left(
  \begin{array}{c}
    \mathbf{S}_k \mathbf{W}_{d_k}\\
    \mathbf{S}'_{k}\mathbf{X}_{{\bf d}}-\sum\limits_{k'\in[0,K)\setminus \{k\}}\mathbf{D}_{k,k'}\mathbf{S}_0\mathbf{W}_{d_{k'}}
 \end{array}
\right).
\end{eqnarray*}
This implies that
\begin{eqnarray}
\label{eq-decoding-2}
\mathbf{W}_{d_k}=\left(
  \begin{array}{c}
    \mathbf{S}_k\\
    \mathbf{S}'_{k}\mathbf{A}_{k}
 \end{array}
\right)^{-1} \left(
  \begin{array}{c}
    \mathbf{S}_k \mathbf{W}_{d_k}\\
    \mathbf{S}'_{k}\mathbf{X}_{{\bf d}}-\sum\limits_{k'\in[0,K)\setminus \{k\}}\mathbf{D}_{k,k'}\mathbf{S}_0\mathbf{W}_{d_{k'}}
 \end{array}
\right).
\end{eqnarray}
\end{remark}

The following example demonstrate how the scheme works.
\begin{example}
\label{exm3}
Let $K=N=6$, $M=3$ and $F=4$. We have $\frac{FM}{N}=2$. For $k=0,1,\ldots$, $5$, define $\mathbf{S}'_k=\mathbf{S}_k$ and $\mathbf{A}_k$ in the following way.
\begin{equation*}
\mathbf{S}_0=\left(
\begin{array}{c}
{\bf e}_0\\ {\bf e}_2
\end{array}
\right)\ \
\mathbf{S}_1=\left(
\begin{array}{c}
{\bf e}_1\\ {\bf e}_3
\end{array}
\right)\ \
\mathbf{S}_2=\left(
\begin{array}{c}
{\bf e}_0\\ {\bf e}_1
\end{array}
\right)\  \
\mathbf{S}_3=\left(
\begin{array}{c}
{\bf e}_2 \\ {\bf e}_3
\end{array}
\right)\ \
\mathbf{S}_4=\left(
\begin{array}{c}
{\bf e}_0+{\bf e}_1\\  {\bf e}_2+{\bf e}_3
\end{array}
\right)\ \
\mathbf{S}_5=\left(
\begin{array}{c}
{\bf e}_0+{\bf e}_2\\ {\bf e}_1+{\bf e}_3
\end{array}
\right)\ \
\end{equation*}
\begin{equation*}
\mathbf{A}_0=\left(
\begin{array}{c}
{\bf e}_1\\ {\bf e}_1\\ {\bf e}_3\\ {\bf e}_3
\end{array}
\right)\ \
\mathbf{A}_1=\left(
\begin{array}{c}
{\bf e}_0\\ {\bf e}_0\\ {\bf e}_2\\ {\bf e}_2
\end{array}
\right)\ \
\mathbf{A}_2=\left(
\begin{array}{c}
{\bf e}_2\\ {\bf e}_3\\ {\bf e}_2\\ {\bf e}_3
\end{array}
\right)\ \
\mathbf{A}_3=\left(
\begin{array}{c}
{\bf e}_0 \\ {\bf e}_1\\ {\bf e}_0\\ {\bf e}_1
\end{array}
\right)\ \
\mathbf{A}_4=\left(
\begin{array}{c}
{\bf e}_0\\ 0\\ {\bf e}_2\\ 0
\end{array}
\right)\ \
\mathbf{A}_5=\left(
\begin{array}{c}
{\bf e}_0\\ {\bf e}_1\\ 0 \\ 0
\end{array}
\right)
\end{equation*}
Here $\mathbf{e}_s$ is a row vector with length $F$ where the $s$th entry is $1$ and other entries are $0$s, $0\leq j<F$.
It is easy to check that \eqref{Eqn rank condition} holds for each integer $k\in [0,6)$.
A $4$-division $(6,3,6)$ coded caching system can be implemented as follows, where $\mathbf{W}_i = (w_i^0, w_i^1, w_i^2, w_i^3)^T, w_i^j \in \mathbb{F}_p$:
\begin{itemize}
\item  \textbf{Placement Phase}:   The contents in each users are
       \begin{align*}
       \mathcal{Z}_0=\left\{\mathbf{S}_0\mathbf{W}_{i}:i\in[0,6)\right\}~~~~ \mathcal{Z}_1=\left\{\mathbf{S}_1\mathbf{W}_{i}:i\in[0,6)\right\}\notag\\
       \mathcal{Z}_2=\left\{\mathbf{S}_2\mathbf{W}_{i}:i\in[0,6)\right\}~~~~ \mathcal{Z}_3=\left\{\mathbf{S}_3\mathbf{W}_{i}:i\in[0,6)\right\}\notag\\
       \mathcal{Z}_4=\left\{\mathbf{S}_4\mathbf{W}_{i}:i\in[0,6)\right\}~~~~ \mathcal{Z}_5=\left\{\mathbf{S}_5\mathbf{W}_{i}:i\in[0,6)\right\}\notag
       \end{align*}
 \item \textbf{Delivery Phase}: Assume the request vector $\mathbf{d}=(0,1,\cdots,5)$. Then the server just broadcasts
 \begin{eqnarray*}
\mathbf{X}_{{\bf d}}=\mathbf{A}_0\mathbf{W}_0+\mathbf{A}_1\mathbf{W}_1+\mathbf{A}_2\mathbf{W}_2+\mathbf{A}_3\mathbf{W}_3+\mathbf{A}_4\mathbf{W}_4+
\mathbf{A}_5\mathbf{W}_5
\end{eqnarray*}
\end{itemize}
Now let us consider user $0$ first. User $0$ can obtain the following content by $\mathbf{S}'_0$ and $\mathbf{X}_{{\bf d}}$.
\begin{eqnarray*}
\mathbf{S}'_0\mathbf{X}_{{\bf d}}&=&\mathbf{S}_0\mathbf{A}_0\mathbf{W}_0+\mathbf{S}_0\mathbf{A}_1\mathbf{W}_1+\mathbf{S}_0\mathbf{A}_2\mathbf{W}_2+\mathbf{S}_0\mathbf{A}_3\mathbf{W}_3+\mathbf{S}_0\mathbf{A}_4\mathbf{W}_4+\mathbf{S}_0\mathbf{A}_5\mathbf{W}_5\\
&=& {{\bf e}_1\choose {\bf e}_3}\mathbf{W}_0+{{\bf e}_0\choose {\bf e}_2}\mathbf{W}_1+{{\bf e}_2\choose {\bf e}_2}\mathbf{W}_2+{{\bf e}_0\choose {\bf e}_0}\mathbf{W}_3+{{\bf e}_0\choose {\bf e}_2}\mathbf{W}_4+{{\bf e}_0\choose 0}\mathbf{W}_5
\end{eqnarray*}
That is,
\begin{eqnarray*}
{{\bf e}_1\choose {\bf e}_3}\mathbf{W}_0&=&\mathbf{S}'_0\mathbf{X}_{{\bf d}}-{{\bf e}_0\choose {\bf e}_2}\mathbf{W}_1-{{\bf e}_2\choose {\bf e}_2}\mathbf{W}_2-{{\bf e}_0\choose {\bf e}_0}\mathbf{W}_3-{{\bf e}_0\choose {\bf e}_2}\mathbf{W}_4-{{\bf e}_0\choose 0}\mathbf{W}_5.
\end{eqnarray*}
Then we could get matrices
\begin{eqnarray*}
\mathbf{D}_{0,1}=\left(\begin{matrix}1&0\\ 0&1\end{matrix}\right),\
\mathbf{D}_{0,2}=\left(\begin{matrix}0&1\\ 0&1\end{matrix}\right),\
\mathbf{D}_{0,3}=\left(\begin{matrix}1&0\\ 1&0\end{matrix}\right),\
\mathbf{D}_{0,4}=\left(\begin{matrix}1&0\\ 0&1\end{matrix}\right),\
\mathbf{D}_{0,5}=\left(\begin{matrix}1&0\\ 0&0\end{matrix}\right)
\end{eqnarray*}
satisfying \eqref{eq-decoding_matrix_1}. So we have
\begin{eqnarray*}
{{\bf e}_1\choose {\bf e}_3}\mathbf{W}_0={w^1_0\choose w^3_0}=\mathbf{S}'_0\mathbf{X}_{{\bf d}}-\mathbf{D}_{0,1}\mathbf{S}_0\mathbf{W}_1-\mathbf{D}_{0,2}\mathbf{S}_0\mathbf{W}_2-\mathbf{D}_{0,3}\mathbf{S}_0\mathbf{W}_3-
\mathbf{D}_{0,4}\mathbf{S}_0\mathbf{W}_4-\mathbf{D}_{0,5}\mathbf{S}_0\mathbf{W}_5.
\end{eqnarray*}
By \eqref{eq-decoding-2}, User $0$ can get
\[\mathbf{W}_0=\left(
\begin{array}{c}
    \mathbf{S}_0 \\
    \mathbf{S}_0\mathbf{A}_{0}
  \end{array}\right)^{-1}
\left(
\begin{array}{c}
    w^0_0 \\
    w^2_0\\
    w^1_0\\
    w^3_0
  \end{array}\right)
= \left(
\begin{array}{c}
    {\bf e}_0 \\
    {\bf e}_2\\
    {\bf e}_1\\
    {\bf e}_3
  \end{array}\right)^{-1}\left(
\begin{array}{c}
    w^0_0 \\
    w^2_0\\
    w^1_0\\
    w^3_0
  \end{array}\right) = \left(
\begin{array}{c}
    w^0_0 \\
    w^1_0\\
    w^2_0\\
    w^3_0
  \end{array}\right).
\]
Similarly we can show that the other users' requests can be all satisfied.
\end{example}

Now we indicate that the linear coded caching system is a generalization of the PDA. We have the following theorem which
is proved in Appendix.

\begin{theorem}
\label{th-PDA-Linear}
The coded caching scheme obtained by a PDA is also linear.
\end{theorem}

The following example gives a demonstration.

\begin{example}
\label{exm4}
It is easy to check that in Example \ref{exam2} the coded caching scheme realized by a $(6,4,2,4)$ PDA can be represented by the following matrices $\mathbf{S}_k=\mathbf{S}'_k$ and $\mathbf{A}_k$, $k=0$, $1$, $\ldots$, $5$.
\begin{eqnarray*}
\mathbf{S}_0=\left(\begin{array}{c}
 {\bf e}_0\\
 {\bf e}_2
\end{array}
\right)\
\mathbf{S}_1=\left(\begin{array}{c}
 {\bf e}_1\\
 {\bf e}_3
\end{array}
\right)\
\mathbf{S}_2=\left(\begin{array}{c}
 {\bf e}_0\\
 {\bf e}_1
\end{array}
\right)\
\mathbf{S}_3=\left(\begin{array}{c}
 {\bf e}_2\\
 {\bf e}_3
\end{array}
\right)\
\mathbf{S}_4=\left(\begin{array}{c}
 {\bf e}_0\\
 {\bf e}_3
\end{array}
\right)\
\mathbf{S}_5=\left(\begin{array}{c}
 {\bf e}_1\\
 {\bf e}_2
\end{array}
\right)
\end{eqnarray*}
\begin{eqnarray*}
\mathbf{A}_0=\left(\begin{array}{cccc}
 0&  1&  0& 0\\
 0&  0& 0&  0\\
 0&  0&  0& 1\\
 0&  0&  0& 0
 \end{array}
\right)\ \
\mathbf{A}_1=\left(\begin{array}{cccc}
 0&  0&  0& 0\\
 1&  0& 0&  0\\
 0&  0&  0& 0\\
 0&  0&  1& 0
 \end{array}
\right)\ \
\mathbf{A}_2=\left(\begin{array}{cccc}
 0&  0& 1& 0\\
 0&  0& 0& 1\\
 0&  0&  0& 0\\
 0&  0&  0& 0
 \end{array}
\right)\\[0.2cm]
\mathbf{A}_3=\left(\begin{array}{cccc}
 0&  0&  0& 0\\
 0&  0& 0&  0\\
 1&  0&  0& 1\\
 0& 1 &  0& 0
 \end{array}
\right)\ \
\mathbf{A}_4=\left(\begin{array}{cccc}
 0&  0& 0& 0\\
 0&  1& 0&  0\\
 0&  0& 1& 1\\
 0&  0&  0& 0
 \end{array}
\right)\ \
\mathbf{A}_5=\left(\begin{array}{cccc}
 1&  0&  0& 0\\
 0&  0& 0&  0\\
 0&  0&  0& 0\\
 0&  0&  0& 1
 \end{array}
\right)
\end{eqnarray*}
It is easy to check that \eqref{Eqn rank condition} holds for each integer $k\in [0,5]$.
\end{example}

\section{A New construction of linear coded caching scheme}
\label{sec-new-construction}

From Theorem \ref{th-main}, we  need to find the three classes of matrices $\mathbf{S}_k$, $\mathbf{A}_k$ and $\mathbf{S}'_k$, $k\in[0,K)$
satisfying the rank conditions.
In this section, we give one concrete  construction. There should be different constructions and one interesting open problem is finding some better
constructions.

Throughout this paper we do not specifically distinguish a matrix and the vector space spanned by its rows if the context is clear. Note that
 \eqref{Eqn rank condition} holds if and only if the following formula holds.
\begin{eqnarray}\label{Eqn space condition}
\left\{
\begin{array}{ll}
\mathbf{S}'_{k}\mathbf{A}_{k'}\subseteq\mathbf{S}_k &\mbox{if}\ \    k\neq k'\\
\mathbf{S}_k+\mathbf{S}'_{k}\mathbf{A}_{k'}=\mathbb{F}_2^F&  \mbox{if}\ \    k= k'
\end{array}
\right.\ \ \  k,k'\in [1,K)
\end{eqnarray}
Here the sum of two subspace $\mathcal{U}$, $\mathcal{V}$ of $\mathbb{F}^F_2$ is defined as $\mathcal{U}+\mathcal{V}=\{{\bf u}+{\bf v}\ |\ {\bf u}\in\mathcal{U},{\bf v}\in\mathcal{V}\}$.

In this section, we will focus on the scheme over $\mathbb{F}_2^F$. For simplicity, we assume that each file has the size of $F$ bits. 
\subsection{Intuitions}
From the proof of Theorem \ref{th-PDA-Linear}, we known that a PDA can be represented by caching matrices, coding matrices and decoding matrices satisfying. So we may construct a linear coded caching schemes by borrowing some similar key structures which are used in constructions
of  PDAs. But we should note here that there should be other methods to construct linear code caching schemes.

Now let us construct caching matrices, coding matrices and decoding matrices satisfying \eqref{Eqn space condition}. For any fixed integers $m$ and $q$, we can denote each integer $s\in[0,q^m)$ by $s=(s_0,\cdots, s_{m-1} )_q$ if $s=\sum_{l=0}^{m-1}s_lq^l$ where $0\leq l<m$ and $s_l\in[0,q)$. We refer to $s=(s_0,\cdots, s_{m-1})_q$ as the {\em $q$-ary representation\/} of $s$.
There are $m$ partitions of $[0,q^m)$, \emph{i.e.}, for each $0\leq u<m$, the $u$th partition is
\begin{eqnarray}\label{Eqn_Partition_3}
\mathcal{V}_{u,v}=\left\{s=(s_0,\cdots, s_{m-1})_q\ \left|\ s_{u}=v, s\in[0,q^m)\right.\right\}, \ \  0\le v<q.
\end{eqnarray}
which are the placement sets to construct PDA in \cite{YCTC}.
It is easy to check that the following formula holds for any integers $u_1$, $u_2$ and $v_1$, $v_2$ with $(u_1,v_1)\neq (u_2,v_2)$.
\begin{eqnarray*}\label{Eqn partition}
\left|\mathcal{V}_{u_1,v_1}\bigcap\mathcal{V}_{u_2,v_2}\right|=\left\{
\begin{array}{ll}
0 &\mbox{if}\ \    u_1=u_2, v_1\neq v_2\\
q^{m-2}&  \mbox{if}\ \    u_1\neq u_2
\end{array}
\right.
\end{eqnarray*}
Let ${\bf e}_s$ be a $q^m$ length row vector where the $s$th entry is $1$ and other entries are $0$s. Define
\begin{eqnarray}\label{Eqn_base vector1}
\mathcal{E}_{u,v}=\{{\bf e}_{s}\ |\ s\in \mathcal{V}_{u,v}\}
\end{eqnarray}
and
\begin{eqnarray}\label{Eqn_sum vector1}
\mathcal{Q}_{u}=\left\{\sum_{s_u=0}^{q-1}{\bf e}_{(s_0,\cdots, s_{m-1})_q}\ \Big|\ s_j\in [0,q), j\neq u \right\}
\end{eqnarray}
where the sum is performed under modulo $q$.
For each $u$, $v$ and $v'$ with $v\neq v'$, define
\begin{eqnarray}
\label{eq-cons1-2}
\mathbf{C}_{u,v,v'}=\left(
\begin{matrix}
\phi_{u,v}(0){\bf e}_{\varphi_{u,v'}(0)}\\
\phi_{u,v}(1){\bf e}_{\varphi_{u,v'}(1)}\\
\ldots\\
\phi_{u,v}(q^m-1){\bf e}_{{\varphi_{u,v'}(q^m-1)}}
\end{matrix}\right)+\left(\label{eq-cons1-2}
\begin{matrix}
\phi_{u,v'}(0){\bf e}_0\\
\phi_{u,v'}(1){\bf e}_1\\
\ldots\\
\phi_{u,v'}(q^m-1){\bf e}_{q^m-1}
\end{matrix}\right)
\end{eqnarray}
and
\begin{eqnarray}\label{eq-cons1-3}
\mathbf{C}_{u,q,v}=\left(
\begin{matrix}
\phi_{u,v}(0){\bf e}_0\\
\phi_{u,v}(1){\bf e}_1\\
\ldots\\
\phi_{u,v}(q^m-1){\bf e}_{q^m-1}
\end{matrix}\right)
\end{eqnarray}
where
\begin{eqnarray}
\phi_{u,v}(s)=\left\{\label{eq-fuct2}
\begin{array}{cc}
1& \hbox{if}\ s\in \mathcal{V}_{u,v}\\
0  &  \hbox{otherwise}\ \ \
\end{array}\right.\ \ \ \ \ \
\varphi_{u,v}(s)=(s_0,\ldots,s_{u-1},v , s_{u+1},\ldots,s_{m-1})_q.
\end{eqnarray}

\

\begin{example}
\label{exam6}
Let $q=3$ and $m=2$. From \eqref{Eqn_base vector1} and \eqref{Eqn_sum vector1} we have
\begin{eqnarray*}\label{Eqn_base 32vector}
\begin{split}
\mathcal{E}_{0,0}=\{{\bf e}_{0},\ {\bf e}_{3},\ {\bf e}_{6}\},& \ \ \ \ \
\mathcal{E}_{0,1}=\{{\bf e}_{1},\ {\bf e}_{4},\ {\bf e}_{7}\},& \ \ \ \ \
\mathcal{E}_{0,2}=\{{\bf e}_{2},\ {\bf e}_{5},\ {\bf e}_{8}\},\\
\mathcal{E}_{1,0}=\{{\bf e}_{0},\ {\bf e}_{1},\ {\bf e}_{2}\},& \ \ \ \ \
\mathcal{E}_{1,1}=\{{\bf e}_{3},\ {\bf e}_{4},\ {\bf e}_{5}\}& \ \ \ \ \
\mathcal{E}_{1,2}=\{{\bf e}_{6},\ {\bf e}_{7},\ {\bf e}_{8}\}
\end{split}
\end{eqnarray*}
and
\begin{eqnarray*}\label{Eqn_sum 32vector}
\begin{split}
&\mathcal{Q}_{0}=\{{\bf e}_{0}+{\bf e}_{1}+{\bf e}_{2},\ {\bf e}_{3}+ {\bf e}_{4}+{\bf e}_{5},\ {\bf e}_{6}+{\bf e}_{7}+{\bf e}_8\}, \\
&\mathcal{Q}_{1}=\{{\bf e}_{0}+{\bf e}_{3}+{\bf e}_{6},\ {\bf e}_{1}+{\bf e}_{4}+{\bf e}_{7},\ {\bf e}_{2}+{\bf e}_{5}+{\bf e}_{8}\}.
\end{split}
\end{eqnarray*}
\end{example}

By \eqref{eq-cons1-2} and \eqref{eq-cons1-3}, we have

\begin{eqnarray}
\label{eq-CM-q=3}
\begin{small}
\begin{split}
&&\mathbf{C}_{0,0,1}=\left(\begin{matrix}{\bf e}_1\\[-0.1cm] {\bf e}_1\\[-0.1cm] 0\\[-0.1cm] {\bf e}_4\\[-0.1cm] {\bf e}_4\\[-0.1cm] 0\\[-0.1cm] {\bf e}_7\\[-0.1cm] {\bf e}_7\\[-0.1cm] 0 \end{matrix}\right),
\mathbf{C}_{0,0,2}=\left(\begin{matrix}{\bf e}_2\\[-0.1cm] 0\\[-0.1cm] {\bf e}_2\\[-0.1cm] {\bf e}_5\\[-0.1cm] 0\\[-0.1cm] {\bf e}_5\\[-0.1cm] {\bf e}_8\\[-0.1cm] 0\\[-0.1cm] {\bf e}_8 \end{matrix}\right),
\mathbf{C}_{0,1,0}=\left(\begin{matrix}{\bf e}_0\\[-0.1cm] {\bf e}_0\\[-0.1cm] 0\\[-0.1cm] {\bf e}_3\\[-0.1cm] {\bf e}_3\\[-0.1cm] 0\\[-0.1cm] {\bf e}_6\\[-0.1cm] {\bf e}_6\\[-0.1cm] 0 \end{matrix}\right),
\mathbf{C}_{1,1,2}=\left(\begin{matrix}0\\[-0.1cm] {\bf e}_2\\[-0.1cm] {\bf e}_2\\[-0.1cm] 0\\[-0.1cm] {\bf e}_5\\[-0.1cm] {\bf e}_5\\[-0.1cm] 0\\[-0.1cm] {\bf e}_8\\[-0.1cm] {\bf e}_8 \end{matrix}\right),
\mathbf{C}_{0,2,0}=\left(\begin{matrix}{\bf e}_{0}\\[-0.1cm] 0\\[-0.1cm] {\bf e}_{0}\\[-0.1cm] {\bf e}_3\\[-0.1cm]0\\[-0.1cm] {\bf e}_3 \\[-0.1cm] {\bf e}_6 \\[-0.1cm] 0\\[-0.1cm] {\bf e}_6 \end{matrix}\right),
\mathbf{C}_{0,2,1}=\left(\begin{matrix}0\\[-0.1cm] {\bf e}_{1}\\[-0.1cm] {\bf e}_1\\[-0.1cm] 0\\[-0.1cm] {\bf e}_4\\[-0.1cm] {\bf e}_4 \\[-0.1cm] 0 \\[-0.1cm] {\bf e}_7\\[-0.1cm] {\bf e}_7 \end{matrix}\right),\\
&&\mathbf{C}_{0,3,0}=\left(\begin{matrix}{\bf e}_{0}\\[-0.1cm] 0\\[-0.1cm] 0\\[-0.1cm] {\bf e}_3\\[-0.1cm] 0\\[-0.1cm] 0 \\[-0.1cm] {\bf e}_6 \\[-0.1cm] 0\\[-0.1cm] 0 \end{matrix}\right),
\mathbf{C}_{0,3,1}=\left(\begin{matrix}0\\[-0.1cm] {\bf e}_{1}\\[-0.1cm] 0\\[-0.1cm] 0\\[-0.1cm] {\bf e}_4\\[-0.1cm] 0 \\[-0.1cm] 0 \\[-0.1cm] {\bf e}_7\\[-0.1cm] 0 \end{matrix}\right),
\mathbf{C}_{1,0,1}=\left(\begin{matrix}{\bf e}_3\\[-0.1cm] {\bf e}_4\\[-0.1cm] {\bf e}_5\\[-0.1cm] {\bf e}_3\\[-0.1cm] {\bf e}_4\\[-0.1cm] {\bf e}_5\\[-0.1cm] 0\\[-0.1cm] 0\\[-0.1cm] 0 \end{matrix}\right),
\mathbf{C}_{1,0,2}=\left(\begin{matrix}{\bf e}_6\\[-0.1cm] {\bf e}_7\\[-0.1cm] {\bf e}_8\\[-0.1cm] 0\\[-0.1cm] 0\\[-0.1cm] 0\\[-0.1cm] {\bf e}_6\\[-0.1cm] {\bf e}_7\\[-0.1cm] {\bf e}_8 \end{matrix}\right),
\mathbf{C}_{1,1,0}=\left(\begin{matrix}{\bf e}_0\\[-0.1cm] {\bf e}_1\\[-0.1cm] {\bf e}_2\\[-0.1cm]{\bf e}_0\\[-0.1cm] {\bf e}_1\\[-0.1cm] {\bf e}_2\\[-0.1cm] 0\\[-0.1cm] 0\\[-0.1cm] 0\end{matrix}\right),
\mathbf{C}_{1,1,2}=\left(\begin{matrix}0\\[-0.1cm] 0\\[-0.1cm] 0\\[-0.1cm] {\bf e}_6\\[-0.1cm] {\bf e}_7\\[-0.1cm] {\bf e}_8\\[-0.1cm] {\bf e}_6\\[-0.1cm] {\bf e}_7\\[-0.1cm] {\bf e}_8\end{matrix}\right),\\
&&\mathbf{C}_{1,2,0}=\left(\begin{matrix}{\bf e}_0 \\[-0.1cm] {\bf e}_1\\[-0.1cm] {\bf e}_2\\[-0.1cm] 0\\[-0.1cm] 0\\[-0.1cm] 0 \\[-0.1cm] {\bf e}_{0}\\[-0.1cm] {\bf e}_{1}\\[-0.1cm] {\bf e}_{2}\end{matrix}\right),
\mathbf{C}_{1,2,1}=\left(\begin{matrix}0 \\[-0.1cm] 0\\[-0.1cm] 0\\[-0.1cm] {\bf e}_3\\[-0.1cm] {\bf e}_4\\[-0.1cm] {\bf e}_5 \\[-0.1cm] {\bf e}_{3}\\[-0.1cm] {\bf e}_{4}\\[-0.1cm] {\bf e}_{5}\end{matrix}\right),
\mathbf{C}_{1,3,0}=\left(\begin{matrix}{\bf e}_0 \\[-0.1cm] {\bf e}_1\\[-0.1cm] {\bf e}_2 \\[-0.1cm] 0\\[-0.1cm] 0\\[-0.1cm] 0\\[-0.1cm] 0\\[-0.1cm] 0\\[-0.1cm] 0\end{matrix}\right),
\mathbf{C}_{1,3,1}=\left(\begin{matrix}0 \\[-0.1cm] 0\\[-0.1cm] 0 \\[-0.1cm] {\bf e}_{3}\\[-0.1cm] {\bf e}_{4}\\[-0.1cm] {\bf e}_{5}\\[-0.1cm] 0\\[-0.1cm] 0\\[-0.1cm] 0\end{matrix}\right)\ \ \ \ \ \ \ \ \ \ \ \ \ \ \ \ \ \ \ \ \ \ \ \ \ \ \ \ \ \ \ \ \ \ \ \ \ \ \ \ \ \ \
\end{split}
\end{small}
\end{eqnarray}

The following result is very useful for our new construction of coded caching scheme.
\begin{lemma}
\label{lem1}
Given positive integers $m$ and $q\geq 2$, sets and matrices in \eqref{Eqn_base vector1}, \eqref{Eqn_sum vector1}, \eqref{eq-cons1-2} and \eqref{eq-cons1-3} satisfy the following conditions for any integers $u_1$, $u_2\in [0,m)$ and any integers $v_1$, $v_2\in [0,q)$, $v_3\in[0,q]$.
\begin{itemize}
\item[(I)] If $u_1\neq u_2$, $\mathcal{E}_{u_1,v_1}\mathbf{C}_{u_2,v_3,v_2}\subseteq\mathcal{E}_{u_1,v_1}$.
\item[(II)] If $v_1=v_3$, $\mathcal{E}_{u_1,v_1}\mathbf{C}_{u_1,v_3,v_2}=\mathcal{E}_{u_1,v_2}$. Otherwise, if $v_1\neq v_3$, $\mathcal{E}_{u_1,v_1}\mathbf{C}_{u_1,v_3,v_2}\subseteq\mathcal{E}_{u_1,v_1}$.
\item[(III)] If $u_1=u_2$ and $v_3=q$, $\mathcal{Q}_{u_1}\mathbf{C}_{u_2,v_3,v_2}=\mathcal{E}_{u_1,v_2}$. Otherwise, if $u_1\neq u_2$ or $v_3\neq q$, $\mathcal{Q}_{u_1}\mathbf{C}_{u_2,v_3,v_2}\subseteq\mathcal{Q}_{u_1}$.
\end{itemize}
\end{lemma}
The proof of Lemma \ref{lem1} is included in Appendix B. For any positive integer $z$ with $1\leq z<q$, let $h=\lfloor \frac{q-1}{q-z}\rfloor$. For any integer $v\in [0,q)$ and $\varepsilon\in [0,h)$, we can define $h$ sets
\begin{eqnarray}
\label{eq-Group1}
\begin{split}
\mathcal{G}_{v,\varepsilon}=\{v+1+\varepsilon(q-z),v+2+\varepsilon(q-z), \ldots, v+(q-z)+\varepsilon(q-z)\}_q,
\end{split}
\end{eqnarray}
Here $A_q=\{a\ (\hbox{mod}\ q)\ |\ a\in A\}$ for any given integer set $A$.

\begin{example}
\label{exam4}

\

\begin{itemize}
\item When $q=3$ and $z=2$, we have $h=\lfloor \frac{q-1}{q-z}\rfloor=2$. From \eqref{eq-Group1}, we have
\begin{eqnarray}
\label{eq-Group-m=3-q=3}
\begin{split}
&\mathcal{G}_{0,0}=\{1\},\ \ \ \mathcal{G}_{0,1}=\{2\},\ \ \ \mathcal{G}_{1,0}=\{2\},\ \ \ \mathcal{G}_{1,1}=\{0\},\ \ \ \mathcal{G}_{2,0}=\{0\},\ \ \ \mathcal{G}_{2,1}=\{1\}
\end{split}
\end{eqnarray}
\item When $q=8$ and $z=6$, we have $h=\lfloor \frac{q-1}{q-z}\rfloor=3$. By \eqref{eq-Group1}, we have
\begin{eqnarray*}
&&\mathcal{G}_{0,0}=\{1,2\},\ \ \ \mathcal{G}_{0,1}=\{3,4\},\ \ \ \mathcal{G}_{0,2}=\{5,6\},\\
&&\mathcal{G}_{1,0}=\{2,3\},\ \ \ \mathcal{G}_{1,1}=\{4,5\},\ \ \ \mathcal{G}_{1,2}=\{6,7\},\\
&&\mathcal{G}_{2,0}=\{3,4\},\ \ \ \mathcal{G}_{1,1}=\{5,6\},\ \ \ \mathcal{G}_{1,2}=\{7,0\},\\
&&\mathcal{G}_{3,0}=\{4,5\},\ \ \ \mathcal{G}_{1,1}=\{6,7\},\ \ \ \mathcal{G}_{1,2}=\{0,1\},\\
&&\mathcal{G}_{4,0}=\{5,6\},\ \ \ \mathcal{G}_{1,1}=\{7,0\},\ \ \ \mathcal{G}_{1,2}=\{1,2\},\\
&&\mathcal{G}_{5,0}=\{6,7\},\ \ \ \mathcal{G}_{1,1}=\{0,1\},\ \ \ \mathcal{G}_{1,2}=\{2,3\},\\
&&\mathcal{G}_{6,0}=\{7,0\},\ \ \ \mathcal{G}_{1,1}=\{1,2\},\ \ \ \mathcal{G}_{1,2}=\{3,4\},\\
&&\mathcal{G}_{7,0}=\{0,1\},\ \ \ \mathcal{G}_{1,1}=\{2,3\},\ \ \ \mathcal{G}_{1,2}=\{4,5\}.
\end{eqnarray*}
\end{itemize}
\end{example}
\subsection{New construction}
\begin{construction}
\label{con1}
Given integers $m\geq 1$, $q\geq 2$ and $z$ with $z<q$, each user $k$ is represented by tuples $(u,v,\varepsilon)$ and $(u,q,\varepsilon)$ in the following for convenience, where $0\leq u<m$, $0\leq v< q$ and $0\leq \varepsilon<\lfloor\frac{q-1}{q-z}\rfloor$. From \eqref{eq-cons1-2}, \eqref{eq-cons1-3} and \eqref{eq-Group1}, for any $(u,v,\varepsilon)$ and $(u,q,\varepsilon)$ we can construct the caching matrices as follows
\begin{eqnarray}
\label{eq-caching1}
\begin{split}
&\mathbf{S}_{u,v,\varepsilon}=\{\mathcal{E}_{u,v'}\ |\ v'\in [0,q)\setminus G_{v,\varepsilon}\},\\
&\mathbf{S}_{u,q,\varepsilon}=\mathcal{Q}_u\bigcup\{\mathcal{E}_{u,v'}\ |\ v'\in [0,q-1)\setminus G_{q-1,\varepsilon}\}.
\end{split}
\end{eqnarray}

The coding matrices are
\begin{eqnarray}
\label{eq-coding1}
\begin{split}
&\mathbf{A}_{u,v,\varepsilon}=\left(\begin{matrix}\mathbf{C}_{u,v,v_0}\\ \mathbf{C}_{u,v,v_1}\\ \vdots \\ \mathbf{C}_{u,v,v_{q-z}} \end{matrix}\right), v_i\in \mathcal{G}_{v,\varepsilon},\ \ \hbox{and}\ \
&\mathbf{A}_{u,q,\varepsilon}=\left(\begin{matrix}\mathbf{C}_{u,q,v'_0}\\ \mathbf{C}_{u,q,v'_1}\\ \vdots \\ \mathbf{C}_{u,q,v'_{q-z}} \end{matrix}\right), v'_i\in \mathcal{G}_{q-1,\varepsilon}.
\end{split}
\end{eqnarray}

And the decoding matrices are
\begin{eqnarray}
\label{eq-decoding1}
\mathbf{S}'_{u,v,\varepsilon}=diag(\underbrace{\mathcal{E}_{u,v},\ldots,\mathcal{E}_{u,v}}_{q-z}), \ \ \ \
\mathbf{S}'_{u,q,\varepsilon}=diag(\underbrace{\mathcal{Q}_{u},\ldots,\mathcal{Q}_{u}}_{q-z}).
\end{eqnarray}
\end{construction}

\begin{example}
\label{exm5}
We use the parameters $q$, $m$ in Example \ref{exam6}.
\begin{itemize}
\item When $z=1$, we have $h=1$. By \eqref{eq-caching1} and \eqref{eq-decoding1} we have caching matrices matrices
\begin{eqnarray*}
\begin{array}{ccc}
\mathbf{S}_{0,0,0}=\mathcal{E}_{0,0}
=\left(\begin{matrix}{\bf e}_0\\[-0.1cm] {\bf e}_3\\[-0.1cm] {\bf e}_6\end{matrix}\right),\ \ \ \ \ \ \ \ \ \ \ \ \ &
\mathbf{S}_{0,1,0}=\mathcal{E}_{0,1}
=\left(\begin{matrix}{\bf e}_1 \\[-0.1cm] {\bf e}_4\\[-0.1cm] {\bf e}_7\end{matrix}\right),\ \ \ \ \ \ \ \ \ \ \ \ \ &
\mathbf{S}_{0,2,0}=\mathcal{E}_{0,2}
=\left(\begin{matrix} {\bf e}_2\\[-0.1cm] {\bf e}_5\\[-0.1cm] {\bf e}_8 \end{matrix}\right), \\[0.5cm]
\mathbf{S}_{1,0,0}=\mathcal{E}_{1,0}
=\left(\begin{matrix}{\bf e}_0\\[-0.1cm] {\bf e}_1\\[-0.1cm] {\bf e}_2\end{matrix}\right),\ \ \ \ \ \ \ \ \ \ \ \ &
\mathbf{S}_{1,1,0}=\mathcal{E}_{1,1}
=\left(\begin{matrix}{\bf e}_3\\[-0.1cm] {\bf e}_4\\[-0.1cm] {\bf e}_5\end{matrix}\right),\ \ \ \ \ \ \ \ \ \ \ \ \ &
\mathbf{S}_{1,2,0}=\mathbf{Q}_{1,2}
=\left(\begin{matrix}{\bf e}_6\\[-0.1cm] {\bf e}_7\\[-0.1cm] {\bf e}_8\end{matrix}\right),\\[0.5cm]
\mathbf{S}_{0,3,0}=\mathcal{Q}_{0,3}=\left(\begin{matrix} {\bf e}_{0}+{\bf e}_{1}+{\bf e}_{2}\\[-0.1cm] {\bf e}_{3}+ {\bf e}_{4}+{\bf e}_{5}\\[-0.1cm] {\bf e}_{6}+{\bf e}_{7}+{\bf e}_8\end{matrix}\right),&
\mathbf{S}_{1,3,0}=\mathcal{Q}_{1,3}=\left(\begin{matrix}{\bf e}_{0}+{\bf e}_{3}+{\bf e}_{6}\\[-0.1cm] {\bf e}_{1}+{\bf e}_{4}+{\bf e}_{7}\\[-0.1cm] {\bf e}_{2}+{\bf e}_{5}+{\bf e}_{8}\end{matrix}\right),&
\end{array}
\end{eqnarray*}
coding matrices
\begin{eqnarray*}
\begin{array}{cccc}
\mathbf{A}_{0,0,0}=\left(\begin{matrix}\mathbf{C}_{0,0,1}\\ \mathbf{C}_{0,0,2} \end{matrix}\right),&
\mathbf{A}_{0,1,0}=\left(\begin{matrix}\mathbf{C}_{0,1,0}\\ \mathbf{C}_{0,1,2}\end{matrix}\right),&
\mathbf{A}_{0,2,0}=\left(\begin{matrix}\mathbf{C}_{0,1,0}\\ \mathbf{C}_{0,1,1} \end{matrix}\right),&
\mathbf{A}_{1,0,0}=\left(\begin{matrix}\mathbf{C}_{1,0,1}\\ \mathbf{C}_{1,0,2} \end{matrix}\right),\\[0.5cm]
\mathbf{A}_{1,1,0}=\left(\begin{matrix}\mathbf{C}_{1,1,0}\\ \mathbf{C}_{1,1,2}\end{matrix}\right),&
\mathbf{A}_{1,2,0}=\left(\begin{matrix}\mathbf{C}_{0,1,0}\\ \mathbf{C}_{0,1,1}\end{matrix}\right),&
\mathbf{A}_{0,3,0}=\left(\begin{matrix}\mathbf{C}_{0,1,0}\\ \mathbf{C}_{0,1,2} \end{matrix}\right),&
\mathbf{A}_{1,3,0}=\left(\begin{matrix}\mathbf{C}_{1,1,0}\\ \mathbf{C}_{1,1,2}\end{matrix}\right)
\end{array}
\end{eqnarray*}
and decoding matrices
\begin{eqnarray*}
\begin{array}{cccc}
\mathbf{S}'_{0,0,0}=\left(\begin{matrix}\mathcal{E}_{0,0}&0\\ 0&\mathcal{E}_{0,0}\end{matrix}\right),&
\mathbf{S}'_{0,1,0}=\left(\begin{matrix}\mathcal{E}_{0,1}&0\\ 0&\mathcal{E}_{0,1}\end{matrix}\right),&
\mathbf{S}'_{0,2,0}=\left(\begin{matrix}\mathcal{E}_{0,2}&0\\ 0& \mathcal{E}_{0,2}\end{matrix}\right),&
\mathbf{S}'_{1,0,0}=\left(\begin{matrix}\mathcal{E}_{1,0}&0\\ 0& \mathcal{E}_{1,0}\end{matrix}\right),\\[0.5cm]
\mathbf{S}'_{1,1,0}=\left(\begin{matrix}\mathcal{E}_{1,1}&0\\ 0& \mathcal{E}_{1,1}\end{matrix}\right),&
\mathbf{S}_{1,2,0}=\left(\begin{matrix}\mathcal{E}_{1,2}&0\\ 0& \mathcal{E}_{1,2}\end{matrix}\right),&
\mathbf{S}'_{0,3,0}=\left(\begin{matrix}\mathcal{Q}_{0}&0\\ 0& \mathcal{Q}_{0}\end{matrix}\right),&
\mathbf{S}_{1,3,0}=\left(\begin{matrix}\mathcal{Q}_{1}&0\\ 0& \mathcal{Q}_{1}\end{matrix}\right).
\end{array}
\end{eqnarray*}
\item When $z=2$, we have $h=2$. By \eqref{eq-caching1} and \eqref{eq-decoding1} we have caching matrices matrices
\begin{eqnarray*}
\begin{array}{cccc}
\mathbf{S}_{0,0,0}=\{\mathcal{E}_{0,0},\mathcal{E}_{0,2}\},\ \ &
\mathbf{S}_{0,0,1}=\{\mathcal{E}_{0,0},\mathcal{E}_{0,1}\},\ \ &
\mathbf{S}_{0,1,0}=\{\mathcal{E}_{0,1},\mathcal{E}_{0,0}\},\ \ &
\mathbf{S}_{0,1,1}=\{\mathcal{E}_{0,1},\mathcal{E}_{0,2}\},\\
\mathbf{S}_{0,2,0}=\{\mathcal{E}_{0,2},\mathcal{E}_{0,1}\},\ \ &
\mathbf{S}_{0,2,1}=\{\mathcal{E}_{0,2},\mathcal{E}_{0,0}\},\ \ &
\mathbf{S}_{1,0,0}=\{\mathcal{E}_{1,0},\mathcal{E}_{1,2}\},\ \ &
\mathbf{S}_{1,0,1}=\{\mathcal{E}_{1,0},\mathcal{E}_{1,1}\},\\
\mathbf{S}_{1,1,0}=\{\mathcal{E}_{1,1},\mathcal{E}_{1,0}\},\ \ &
\mathbf{S}_{1,1,1}=\{\mathcal{E}_{1,1},\mathcal{E}_{1,2}\},\ \ &
\mathbf{S}_{1,2,0}=\{\mathcal{E}_{1,2},\mathcal{E}_{1,1}\},\ \ &
\mathbf{S}_{1,2,1}=\{\mathcal{E}_{1,2},\mathcal{E}_{1,0}\},\\
\mathbf{S}_{0,3,0}=\{\mathcal{Q}_{0,3},\mathcal{E}_{0,1}\},\ \ &
\mathbf{S}_{0,3,1}=\{\mathcal{Q}_{0,3},\mathcal{E}_{0,0}\},\ \ &
\mathbf{S}_{1,3,0}=\{\mathcal{Q}_{1,3},\mathcal{E}_{1,1}\},\ \ &
\mathbf{S}_{1,3,1}=\{\mathcal{Q}_{1,3},\mathcal{E}_{1,0}\},
\end{array}
\end{eqnarray*}
coding matrices
\begin{eqnarray*}
\begin{array}{cccc}
\mathbf{A}_{0,0,0}=\mathbf{C}_{0,0,1},\ \ &
\mathbf{A}_{0,0,1}=\mathbf{C}_{0,0,2},\ \ &
\mathbf{A}_{0,1,0}=\mathbf{C}_{0,1,0},\ \ &
\mathbf{A}_{0,1,1}=\mathbf{C}_{0,1,2},\\
\mathbf{A}_{0,2,0}=\mathbf{C}_{0,1,0},\ \ &
\mathbf{A}_{0,2,1}=\mathbf{C}_{0,1,1},\ \ &
\mathbf{A}_{1,0,0}=\mathbf{C}_{1,0,1},\ \ &
\mathbf{A}_{1,0,1}=\mathbf{C}_{1,0,2},\\
\mathbf{A}_{1,1,0}=\mathbf{C}_{1,1,0},\ \ &
\mathbf{A}_{1,1,1}=\mathbf{C}_{1,1,2},\ \ &
\mathbf{A}_{1,2,0}=\mathbf{C}_{0,1,0},\ \ &
\mathbf{A}_{1,2,1}=\mathbf{C}_{0,1,1},\\
\mathbf{A}_{0,3,0}=\mathbf{C}_{0,1,0},\ \ &
\mathbf{A}_{0,3,1}=\mathbf{C}_{0,1,1},\ \ &
\mathbf{A}_{1,3,0}=\mathbf{C}_{1,1,0},\ \ &
\mathbf{A}_{1,3,1}=\mathbf{C}_{1,1,1}
\end{array}
\end{eqnarray*}
and decoding matrices
\begin{eqnarray*}
\begin{array}{cccc}
\mathbf{S}'_{0,0,0}=\mathbf{S}'_{0,0,1}=\mathcal{E}_{0,0},\ \ &
\mathbf{S}'_{0,1,0}=\mathbf{S}'_{0,1,1}=\mathcal{E}_{0,1},\ \ &
\mathbf{S}'_{0,2,0}=\mathbf{S}'_{0,2,1}=\mathcal{E}_{0,2},\ \ &
\mathbf{S}'_{1,0,0}=\mathbf{S}'_{1,0,1}=\mathcal{E}_{1,0},\\
\mathbf{S}'_{1,1,0}=\mathbf{S}'_{1,1,1}=\mathcal{E}_{1,1},\ \ &
\mathbf{S}_{1,2,0}=\mathbf{S}_{1,2,1}=\mathcal{E}_{1,2},\ \ &
\mathbf{S}'_{0,3,0}=\mathbf{S}'_{0,3,1}=\mathcal{Q}_{0},\ \ &
\mathbf{S}_{1,3,0}=\mathbf{S}_{1,3,0}=\mathcal{Q}_{1}.
\end{array}
\end{eqnarray*}
\end{itemize}
\end{example}
Based on Construction \ref{con1}, the following result can be obtained.
\begin{theorem}
\label{th-main-R}
For any positive integers $q$, $z$, $m$ with $q\geq2$ and $z<q$, there exists a $q^m$-division coded caching scheme with parameters $K=m(q+1)\lfloor\frac{q-1}{q-z}\rfloor $, $\frac{M}{N}=\frac{z}{q}$ and rate $R=q-z$. The operation is over the finite filed $\mathbb{F}_2$.
\end{theorem}The proof of Theorem \ref{th-main-R} is included in Appendix C.

\begin{remark}

\

\begin{itemize}
\item Comparing with the uncoded caching strategy, our new scheme adds some computation. However the additional computing
 is very limited since by \eqref{eq-caching1}, \eqref{eq-coding1} and \eqref{eq-decoding1} there are just $m$ caching matrices each
  of which has exactly $q$ entries containing $1$ in each row, and the other matrices has at most one entry containing $1$s in each row.
\item For the fixed user number $K$, when $\frac{M}{N}=\frac{1}{q}$ the following statements hold from Table \ref{tab-main}.
\begin{itemize}
\item  Comparing with MN scheme, we have
\begin{eqnarray*}
\label{eq-com-F-z=1}
\frac{F}{F_{MN}}&\approx\frac{\sqrt{2\pi K(q-1)}}{q}e^{\frac{K}{q+1}\ln q-\frac{K}{q}\left(\ln q+(q-1)\ln \frac{q}{q-1} \right)}=\frac{\sqrt{2\pi K(q-1)}}{q}e^{-\frac{K}{q}\left(\frac{1}{q+1}\ln q+(q-1)\ln \frac{q}{q-1} \right)}
\end{eqnarray*}
\begin{eqnarray*}
\label{eq-com-R-z=1}
\frac{R}{R_{MN}}&=1+\frac{q}{K}.
\end{eqnarray*}
This implies that $\frac{R}{R_{MN}}\rightarrow 1$ when $K\rightarrow \infty$. However, the packet number $F$ of the scheme in Theorem \ref{th-main-R} reduces significantly.

\item Comparing with the scheme in Lemma \ref{le-Yan}, i.e., the second row of Table \ref{tab-main}.
  The scheme is of $F=q^{\frac{K}{q}-1}$ and $R=q-1$ in Lemma \ref{le-Yan}. Our scheme has the packet number $F=q^{\frac{K}{q+1}}$ and the same rate. This implies the ratio of packet number in our scheme to that of the scheme in Lemma \ref{le-Yan} is  $q^{1-\frac{K}{q(q+1)}}$. When $K$
  is very large, the ratio will close to zero.
\end{itemize}
\end{itemize}
\end{remark}
\begin{example}
Assume that $\frac{M}{N}=\frac{1}{2}$. Using Lemmas \ref{le-MN}, \ref{le-Yan} and Theorem \ref{th-main-R}, the following Table shows
the difference of the results.
\begin{table}[!htbp]
\centering
\caption{The schemes obtained from Lemmas \ref{le-MN}, \ref{le-Yan} and Theorem \ref{th-main-R} when $\frac{M}{N}=\frac{1}{2}$ \label{tab-compare}}
\begin{tabular}{|c|c|c|c|c|c|c|}
\hline
\multicolumn{1}{|c|}{ \multirow{2}*{$K$} }&
 \multicolumn{2}{|c|}{ \multirow{1}*{Lemma \ref{le-MN}}} &
 \multicolumn{2}{|c|}{ \multirow{1}*{Lemma \ref{le-Yan}}} &
 \multicolumn{2}{|c|}{ \multirow{1}*{Theorem \ref{th-main-R}}}\\
\cline{2-7}
  &R&F&R&F&R&F\\ \hline
12&0.8571 &924       &1&32    &1&16\\ \hline
18&0.9    &48620     &1&256   &1&64\\  \hline
24&0.9231 &2704156   &1&2048  &1&256\\ \hline
30&0.9375 &155117520 &1&16384 &1&1024\\ \hline
36&0.9474 &9075135300&1&131072&1&4096\\ \hline
\end{tabular}
\end{table}
\end{example}

\section{Concatenating construction}
\label{sec-concatenating}
In order to implement a coded caching scheme subject to a certain subpacketization level for any $K$ users, Cheng et al., in \cite{CJYW} proposed a concatenating construction for PDAs, which can be regarded as a generalization of grouping algorithm in \cite{SJTLD}. In this section,
 we shall use a similar idea to
  propose a concatenating construction of linear coded caching schemes which can be used for
   any $K$ users with low subpacketization level.
   The matrices constructed in \cite{CJYW} are useful. For the ease of convenience, we use the following notations.
\begin{itemize}
\item ${\bf J}_{F\times K}$ denotes a matrix with $F$ rows and $K$ columns where all entries are $1$s.   The matrix ${\bf J}_{F\times 1}$ is denoted by ${\bf J}_F$. Sometimes ${\bf J}_F$ is written as ${\bf J}$ if there is no need to emphasize the parameter $F$.
\item Given an array $\mathbf{P}=(p_{j,k})$, $0\leq j<F$, $0\leq k<K$, with alphabet $\{0,1,\ldots,S-1\}\bigcup\{*\}$, define
 $\mathbf{P}+a=(p_{j,k}+a)$ and $a\mathbf{P}=(ap_{j,k})$ for any integer $a$ where $*+a=*$ and $a*=*$.
\item
$\langle a\rangle_b$ denotes the least non-negative residue of $a$ modulo $b$ for any positive integers $a$ and $b$.
 If $\langle a\rangle_b=0$, we write $b\mid a$. If $\langle a\rangle_b\neq0$, we write $b\nmid a$.
\end{itemize}

\

\begin{construction}(\cite{CJYW})
\label{con2}
For any positive integers $K_1$ and $K_2$ with $K_2\leq K_1$, let $h_1=\frac{K_1}{gcd(K_1,K_2)}$ and $h_2=\frac{K_2}{gcd(K_1,K_2)}$.
\begin{itemize}
\item Define an $h_1\times (K_1+K_2)$ matrix
\begin{small}
\begin{eqnarray}
\label{eq-lable-matrix}
\mathbf{A}=\left(\begin{array}{cccc|cccc }
0 & 1 & \ldots & K_{1}-1 & 0         & 1           & \ldots & K_{2}-1 \\
0 & 1 & \ldots & K_{1}-1 &K_{2}        & K_{2}+1       & \ldots & 2K_{2}-1\\
0 & 1 & \ldots & K_{1}-1 &2K_{2}       & 2K_{2}+1       & \ldots & 3K_{2}-1\\
  &   & \ldots &       &           &             & \ldots &          \\
0 & 1 & \ldots & K_{1}-1 &(h_1-1)K_{2} &(h_1-1)K_{2}+1 & \ldots & h_1K_{2}-1
\end{array}\right),
\end{eqnarray}
\end{small}where all the operations are performed modulo $K_{1}$.
\item Define an $h_1\times (K_{1}+K_{2})$ array $\mathbf{B} = (b_{j,k})$, $0 \leq j < h_1$, $0 \leq k < K_{1}+K_{2}$, with  entries
\begin{eqnarray}\label{eq-recursive1}
b_{j,k}=\left\{
\begin{array}{ll}
j, & \textrm{if}~k\in [0, K_{1}+K_{2}) \setminus \mathcal{A}_j,\\
h_1+\lfloor \frac{jK_{2}+\langle k-jK_{2}\rangle_{K_{1}}}{K_{1}}\rfloor,  & \textrm{if}~k\in \mathcal{A}_j,\\
\end{array}
\right.
\end{eqnarray}
where $\mathcal{A}_j=\{\langle jK_{2}\rangle_{K_{1}},\langle jK_{2}+1\rangle_{K_{1}},\ldots,\langle jK_{2}+K_{2}-1\rangle_{K_{1}}\}$.
\end{itemize}
\end{construction}

\begin{example}\rm
\label{ex-A-B}
Assume that $K_2=2$. Then we have $h_1=3$, $h_2=2$. By \eqref{eq-lable-matrix} and \eqref{eq-recursive1}, we can obtain the following arrays.
\begin{eqnarray*}
\label{eq-lable-matrix-he}
\mathbf{A}=\left(\begin{array}{ccccc}
0 & 1 & 2 &0 & 1\\
0 & 1 & 2 &2 & 0\\
0 & 1 & 2 &1 & 2
\end{array}\right)\ \ \ \ \ \
\mathbf{B}=\left(\begin{array}{ccccc}
3 & 3 & 0 &0 & 0\\
4 & 1 & 3 &1 & 1\\
2 & 4 & 4 &2 & 2
\end{array}\right)
\end{eqnarray*}
\end{example}
The following statement was proposed in \cite{CJYW}. However, we include here a short proof for the reader's convenience.
\begin{lemma}(\cite{CJYW})\label{le-B-Property}
All the entries of each column in matrix $\mathbf{B}$ defined in Construction \ref{con2} are difference.
\end{lemma}
\begin{proof}
It is easy to check that the statement always holds when $K\in [K_1,K_1+K_2)$. Now let us consider the case $k\in [0,K_1)$. For any integers $j$ and $j'$, we assume that $j'>j$. Clearly the hypothesis $b_{j,k}=b_{j',k}\in [0,h_1)$ does happen.

Assume that $b_{j,k}=b_{j',k}\in [h_1,h_1+h_2)$. By \eqref{eq-recursive1}, we have $k\in \mathcal{A}_j \bigcap \mathcal{A}_{j'}$. Let $x_k=\langle k-jK_2\rangle_{K_1}$ and $x'_{k}=\langle k-{j'}K_2\rangle_{K_1}$.
Then $k+K_1z=jK_2+x_k$ and $k+K_1z'=j'K_2+x'_k$ hold for some integers $z$ and $z'$. We claim that $z$ and $z'$ are non-negative integers. We first consider $z$. When $k\geq jK_2$, we have that $x_k=\langle k-jK_2\rangle_{K_1}= k-jK_2$ since $k\in [0,K_1)$. So we have
$x_k+jK_2=k+0K_1$. This implies that $z=0$. When $k< jK_2$, there exists a nonnegative integer $y$ such that $-K_1< k+K_1y-jK_2< 0$. Then we have that $x_k=\langle k-jK_2\rangle_{K_1}= K_1+ k+ K_1y -jK_2$. So we have
$x_k+jK_2=k+K_1(y+1)$. Clearly $z=y+1\geq 1$. Similarly we can prove that $z'$ is a non-negative integer too.  If $z'=z$, then
  $\mathcal{A}_j\bigcap \mathcal{A}_{j'}=\emptyset$, a contradiction to our hypothesis. So $z'>z$ always holds. Then
\begin{eqnarray*}
b_{j,k}&=&h_1+\left\lfloor \frac{jK_2+x_k}{K_1}\right\rfloor=h_1+\left\lfloor z+\frac{k}{K_1}\right\rfloor=h_1+z,\\  b_{j',k}&=&h_1+\left\lfloor \frac{j'K_2+x'_k}{K_1}\right\rfloor=h_1+\left\lfloor z'+\frac{k}{K_1}\right\rfloor=h_1+z'.
\end{eqnarray*}
This implies $b_{j,k}\neq b_{j',k}$, a contradiction to our assumption $b_{j,k}= b_{j',k}$.
\end{proof}

Based on the $k$th column of $\mathbf{B}$ define in Construction \ref{con2}, we can define  another $(h_1+h_2)\times h_1$ matrix
\begin{eqnarray}
\label{eq-Sign}
\Gamma(\mathbf{B},k)=\left(\begin{matrix}\chi_{b_{0, k_1}}(0)&\cdots &\chi_{b_{h_1-1, k_1}}(0)\\
\chi_{b_{0, k_1}}(1)&\cdots &\chi_{b_{h_1-1, k_1}}(1)\\
\vdots& \ddots  &\vdots\\
\chi_{b_{0,k_1}}(h_1+h_2-2)&\cdots &\chi_{b_{h_1-2, k_1}}(h_1+h_2-1)\\
\chi_{b_{0,k_1}}(h_1+h_2-1)&\cdots &\chi_{b_{h_1-1, k_1}}(h_1+h_2-1)\end{matrix}\right)
\end{eqnarray}where
\begin{eqnarray*}
\chi_{k}(x)=\left\{\label{eq-sign-fuct2}
\begin{array}{cc}
1& \hbox{if}\ k=x\\
0  &  \hbox{otherwise}\ \ \
\end{array}\right.
\end{eqnarray*}
\begin{example}
\label{ex-sign-B-k}
We use the matrix $\mathbf{B}$ in Example \ref{ex-A-B}. When $k=0,1,2,3$, the following matrices can be obtained by \eqref{eq-Sign}.
\begin{eqnarray*}
\Gamma(\mathbf{B},0)=\left(
\begin{matrix}
0&0&0\\
0&0&0\\
0&0&1\\
1&0&0\\
0&1&0
\end{matrix}
\right)\ \ \ \
\Gamma(\mathbf{B},1)=\left(
\begin{matrix}
0&0&0\\
0&1&0\\
0&0&0\\
1&0&0\\
0&0&1
\end{matrix}
\right)\ \ \ \
\Gamma(\mathbf{B},2)=\left(
\begin{matrix}
1&0&0\\
0&0&0\\
0&0&0\\
0&1&0\\
0&0&1
\end{matrix}
\right)\ \ \ \
\Gamma(\mathbf{B},2)=\left(
\begin{matrix}
1&0&0\\
0&1&0\\
0&0&1\\
0&0&0\\
0&0&0
\end{matrix}
\right)
\end{eqnarray*}
\end{example}
From Lemma \ref{le-B-Property} we have that each row and each column of $\Gamma(\mathbf{B},k)$ has at most one entry containing $1$, others containing $0$. So the following result holds.
\begin{corollary}
\label{co-h1square}
Assume that $\mathbf{B}$ is a matrix generated by \eqref{eq-recursive1}. Then for each integer $k\in[0,K_1+K_2)$, the matrix $\mathbf{H}=\Gamma(\mathbf{B},k)^T\cdot\Gamma(\mathbf{B},k)$ is a permutation matrix with size $h_1\times h_1$.
\end{corollary}

\begin{construction}
\label{con3}
Given an $F$-division $(K_1,M,N)$ linear coded caching scheme with three classes of matrices $\mathbf{S}_k$, $\mathbf{A}_k$ and $\mathbf{S}'_k$, $k\in[0,K)$, and rate $R$, for each $k_1\in [0,K_1)$ and $k_2\in [K_1,K_1+K_2)$, define $h_1Z\times h_1F$ caching matrices
\begin{eqnarray}
\label{eq-re-caching}
\overline{\mathbf{S}}_{k_1}=diag(\ \underbrace{\mathbf{S}_{k_1},\ldots,\mathbf{S}_{k_1}}_{h_1}\ ), \ \ \ \
\overline{\mathbf{S}}_{k_2}=diag(\ \underbrace{\mathbf{S}_{a_{k_2,0}},\ldots,\mathbf{S}_{a_{k_2,h_1-1}}}_{h_1}\ ),
\end{eqnarray}
$(h_1+h_2)S\times h_1F$ coding matrices
\begin{eqnarray}
\label{eq-re-coding}
&\overline{\mathbf{A}}_{k_1}=\Gamma(\mathbf{B},k_1)\otimes\mathbf{A}_{k_1},\ \ \ \ \
\overline{\mathbf{A}}_{k_2}=diag(\mathbf{A}_{a_{k_2,0}},\ldots,\mathbf{A}_{a_{k_2,h_1-1}},
\underbrace{\mathbf{0},\ldots,\mathbf{0}}_{h_2})
\end{eqnarray}
and $h_1(F-Z)\times (h_1+h_2)RF$ decoding matrices
\begin{eqnarray}
\label{eq-re-decoding}
\overline{\mathbf{S}'}_{k_1}=\Gamma(\mathbf{B},k_1)^{T}\otimes\mathbf{S}'_{k_1},\ \ \ \ \
\overline{\mathbf{A}}_{k_2}=diag(\mathbf{S}'_{a_{k_2,0}},\ldots,\mathbf{S}'_{a_{k_2,h_1-1}})
\end{eqnarray}
where ``$\otimes$" is Kronecker product. That is, for any two matrix $\mathbf{C}=(c_{i,j})$ and $\mathbf{D}$, $\mathbf{C}\otimes \mathbf{D}=(c_{i,j}\mathbf{D})$.
\end{construction}
\begin{lemma}
\label{le-fundamental recursive}
If there exists an $F$-division $(K_1,M,N)$ linear coded caching scheme with rate $R$, then there exists an $h_1F$-division $(K_1+K_2,M,N)$ linear coded caching scheme with $(1+\frac{h_2}{h_1})R$ for any positive integer $K_2\leq K_1$, where $h_1=\frac{K_{1}}{gcd(K_{1},K_{2})}$ and $h_2=\frac{K_{2}}{gcd(K_{1},K_{2})}$.
\end{lemma}
\begin{proof} Given an $F$-division $(K_1,M,N)$ linear coded caching scheme, its caching, coding and decoding matrices are $\mathbf{S}_{k_1}$, $\mathbf{A}_{k_1}$ and $\mathbf{S}'_{k_1}$, $k_1\in [0,K_1)$ respectively. For any positive integer $K_2\leq K_1$, from Constructions \ref{con2} and \ref{con3} we have three classes of matrices $\overline{\mathbf{S}}_{k}$, $\overline{\mathbf{A}}_{k}$ and $\overline{\mathbf{S}'}_{k}$, $k\in [0,K_1+K_2)$. We claim that they satisfy the condition \eqref{Eqn rank condition} in Theorem \ref{th-main}.
\begin{itemize}
\item When $k,k'\in [0,K_1)$, from Corollary \ref{co-h1square} we have
\begin{eqnarray*}
\mathrm{rank}\left(
  \begin{array}{c}
    \overline{\mathbf{S}}_k \\
    \overline{\mathbf{S}'}_{k}\overline{\mathbf{A}}_{k'}
 \end{array}
\right)&=&\mathrm{rank}\left(\begin{matrix}
\begin{array}{ccc}
\mathbf{S}_{k} & & \\
            &\ddots & \\
      & &\mathbf{S}_{k}
\end{array}\\ \hline \\[-0.2cm]
\Gamma(\mathbf{B},k)^{T}\otimes\mathbf{S}'_{k}\cdot
\Gamma(\mathbf{B},k')\otimes\mathbf{A}_{k'}
\end{matrix}\right)
\leq\mathrm{rank}\left(\begin{matrix}
\begin{array}{ccc}
\mathbf{S}_{k} & & \\
            &\ddots & \\
      & &\mathbf{S}_{k}
\end{array}\\ \hline \\[-0.2cm]
\mathbf{H}\otimes \mathbf{S}'_{k}\mathbf{A}_{k'}
\end{matrix}\right) \\
&=&\mathrm{rank}\left(
\begin{array}{ccc}
\mathbf{S}_{k} & & \\
            &\ddots & \\
      & &\mathbf{S}_{k}\\ \hline
\mathbf{S}'_{k}\mathbf{A}_{k} & & \\
            &\ddots & \\
      & &\mathbf{S}'_{k}\mathbf{A}_{k'}
\end{array}\right)
=\mathrm{rank}\left(
\begin{array}{ccc}
\mathbf{S}_{k} & & \\
\mathbf{S}'_{k}\mathbf{A}_{k'} & & \\ \hline
            &\ddots & \\ \hline
      & &\mathbf{S}_{k}\\
      & &\mathbf{S}'_{k}\mathbf{A}_{k'}
\end{array}\right)\\
&=&h_1\cdot\mathrm{rank}\left(
\begin{array}{c}
\mathbf{S}_{k} \\
\mathbf{S}'_{k}\mathbf{A}_{k'}
\end{array}\right)\\
&=&\left\{\begin{array}{lll}h_1F,&\mbox{if}\ \    k=k'
\\
h_1\frac{FM}{N},&\mbox{otherwise}\end{array}\right.  \mbox{~~for~~}0\le k'< K_1.
\end{eqnarray*}

\item When $k\in [0,K_1)$ and $k'\in [K_1,K_1+K_2)$, from Corollary \ref{co-h1square} we have
\begin{eqnarray*}
\mathrm{rank}\left(
  \begin{array}{c}
    \overline{\mathbf{S}}_k \\
    \overline{\mathbf{S}'}_{k}\overline{\mathbf{A}}_{k'}
 \end{array}
\right)&=&\mathrm{rank}\left(\begin{matrix}
\begin{array}{ccc}
\mathbf{S}_{k} & & \\
            &\ddots & \\
      & &\mathbf{S}_{k}
\end{array}\\ \hline \\[-0.2cm]
\Gamma(\mathbf{B},k)^{T}\otimes\mathbf{S}'_{k}\cdot
\overline{\mathbf{A}}_{k'}
\end{matrix}\right)
\leq\mathrm{rank}\left(\begin{matrix}
\begin{array}{ccc}
\mathbf{S}_{k} & & \\
            &\ddots & \\
      & &\mathbf{S}_{k}
\end{array}\\ \hline \\[-0.2cm]
\mathbf{H}\otimes \mathbf{S}'_{k}\mathbf{A}_{k'}
\end{matrix}\right) \\
&=&h_1\cdot\mathrm{rank}\left(
\begin{array}{c}
\mathbf{S}_{k} \\
\mathbf{S}'_{k}\mathbf{A}_{k}
\end{array}\right)\\
&=&h_1\frac{FM}{N}.
\end{eqnarray*}
\item Similar to the above discussions, it is easy to check that \eqref{Eqn rank condition} always holds for any integers $k\in [K_1,K_1+K_2)$ and $k'\in [0,K_1+K_2)$.
\end{itemize}
\end{proof}
\begin{theorem}
\label{th-fundamental recursive}
Given an $F$-division $(K_1,M,N)$ linear coded caching scheme with rate $R$, for any $K>K_1$, there exists an $h_1F$-division $(K,M,N)$ linear coded caching scheme with rate $\frac{K}{K_1}R$, where $h_1=\frac{K_{1}}{gcd(K_{1},K)}$.
\end{theorem}
\begin{proof}
Assume that there exists an $F$-division $(K_1,M,N)$ linear coded caching scheme with rate $R$. Let $K_2 = K -K_1$. We have
 $h_1=\frac{K_1}{gcd(K_1,K_2)} = \frac{K_1}{gcd(K_1,K)}$. Now let us consider the value of $K_2$ as follows.
\begin{itemize}
\item If $K_2\leq K_1$, from Lemma
\ref{le-fundamental recursive}, there exists an $h_1F$-division $(K_1+K_2,M,N)$, i.e., $(K,M,N)$, linear coded caching scheme with $(1+\frac{h_2}{h_1})R=\frac{K}{K_1}R$, where $h_1=\frac{K_{1}}{gcd(K_{1},K_{2})}$ and $h_2=\frac{K_{2}}{gcd(K_{1},K_{2})}$.

\item If $K_2>K_1$, let $m=\lfloor \frac{K_2}{K_1}\rfloor$ and $K_2'=K_2-mK_1$. Clearly $m>1$ always holds. Now we  consider the first $(m+1)K_1$ users. We split these $(m+1)K_1$ users into $m+1$ groups such that each group has $K_1$ users, then use an $F$-division $(K_1,M,N)$ linear coded caching scheme with rate $R$ for these groups respectively. Clearly the total rate is $(m+1)R$. Then we have an $F$-division $((m+1)K_1,M,N)$ linear coded caching scheme with rate $(m+1)R$. Clearly when $K'_2=0$, our statement holds. When $K_2'\neq 0$, based on our $F$-division $((m+1)K_1,M,N)$ linear coded caching scheme we have a new $h_1F$-division $((m+1)K_1+K'_2,M,N)$, i.e., $(K,M,N)$ linear coded caching scheme with rate $(1+\frac{h_2}{h_1})(m+1)R$, where $h_1=\frac{(m+1)K_{1}}{gcd((m+1)K_{1},K'_{2})}$ and $h_2=\frac{K'_{2}}{gcd((m+1)K_{1},K_{2})}$. It is easy to check that
    $$(1+\frac{h_2}{h_1})(m+1)R=\frac{(m+1)K_1+K'_2}{(m+1)K_1}(m+1)R=\frac{K}{K_1}R.$$
\end{itemize}
The proof is completed.
\end{proof}

\section{Constructions from minimum storage regenerating codes}
\label{sec_MSRC}
We notice that the format of the condition \eqref{Eqn rank condition} is very similar to the necessary condition of the optimal bandwidth for minimum storage regenerating (MSR) code. Minimum storage regenerating (MSR) codes was introduced in \cite{DGWWR} for distributed storage systems recently. It is interesting to know if we can obtain some linear coded caching schemes from the constructions of MSR codes.
That is the purpose of this section.

Let us introduce the module and requirements of MSR codes briefly. Assume that a file of size $\mathcal{M}=KF$ denoted by
the column vector ${\bf W}\in\mathbf{F}_p^{K F}$ is partitioned into $k$ parts ${\bf W}=\{{\bf W}_0, {\bf W}_1, \cdots, {\bf W}_{K-1}\}$, each of size $ F$,
where $p$ is a prime power. We encode $\mathbf{W}$ using an $(n=K+r,K)$ MSR code $\mathcal{C}$ and  store it
across $K$ systematic and
$r$ parity storage nodes. Precisely, the first $K$  (systematic) nodes store the file parts ${\bf W}_0$, ${\bf W}_1$, $\ldots$, ${\bf W}_{K-1}$ in an uncoded form respectively,
and the parity nodes store linear combinations of ${\bf W}_0$, ${\bf W}_1$, $\ldots$, ${\bf W}_{K-1}$. Without loss of generality, it is assumed that the nodes $K+x$ stores ${\bf W}_{K+x}=\sum\limits_{i=0}^{K-1}\mathbf{A}_{x,i}{\bf W}_i$ for each $x\in[0,r)$ and some $ F\times F$ matrices $\mathbf{A}_{x,0},\cdots,\mathbf{A}_{x,K-1}$ over $\mathbf{F}_p$, where the {nonsingular} matrix $A_{x,i}$ is called the \emph{encoding matrix} for the $i$th systematic node, $0\le i< K$. { Here the nonsingular property of encoding matrix is necessary to guarantee the resiliency to any $n-k$ node failures. This is also called maximum distance separable (MDS) property.} Table \ref{MSR_Model} illustrates the structure of a $(K+r,K)$ MSR code which has been widely studied.
\begin{table}[htbp]
\begin{center}
\caption{Structure of a $(K+r,K)$ MSR code\label{MSR_Model}}
\begin{tabular}{|c|c|}
\hline
Systematic node & Systematic data \\
\hline
0 & ${\bf W}_0$ \\
\hline
\vdots & \vdots \\
\hline
$K-1$ & ${\bf W}_{K-1}$ \\
\hline
Parity node & Parity data \\
\hline
0 & ${\bf W}_{K}=\mathbf{A}_{0,0}{\bf W}_0+ \cdots+ \mathbf{A}_{K-1,0}{\bf W}_{K-1}$ \\
\hline
1 & ${\bf W}_{K+1}=\mathbf{A}_{0,1}{\bf W}_0+\cdots + \mathbf{A}_{0,1}{\bf W}_{K-1}$ \\
\hline
\vdots & \vdots \\
\hline
r-1 & ${\bf W}_{K+r-1}=\mathbf{A}_{0,r-1}{\bf W}_0+\cdots + \mathbf{A}_{0,r-1}{\bf W}_{K-1}$ \\
\hline
\end{tabular}
\end{center}
\end{table}

Clearly the code is uniquely defined by the matrix
\begin{eqnarray}\label{matrix code}
\mathcal{C}=\left(\begin{array}{ccc}
\mathbf{A}_{0,0} & \ldots &\mathbf{A}_{0,K-1}\\
\vdots  &\ddots  & \vdots\\
\mathbf{A}_{r-1,0} &\ldots  &\mathbf{A}_{r-1,K-1}
\end{array}\right).
\end{eqnarray}


\begin{lemma}(\cite{TWB})
\label{optimal bandwith}
When one information node fails, the code defined in \eqref{matrix code} has optimal repairing bandwidth if there exist subspaces $S_{i,m}$, $\ldots$, $S_{i,K}$ each of dimension $1/r$, such that for any $k$, $k'\in[0,K)$ and $i\in [0,r)$
\begin{eqnarray}\label{property MSR}
\mathrm{rank}\left(
  \begin{array}{c}
    S_{0,k}A_{0,k'} \\
    \vdots\\
    S_{r-1,k}A_{r-1,k'}
 \end{array}
\right)=\left\{\begin{array}{lll} F,&\mbox{if}\ \    k=k'
\\
\frac{ F}{r},&\mbox{otherwise}\end{array}\right.
\end{eqnarray}
\end{lemma}

From aforemention introduction, the following relationship between coded caching scheme and minimum storage regenerating code can be obtained.
\begin{theorem}\label{thm2}
Given a $(n=K+r,K)$ MSR code $\mathcal{C}$ defined in \eqref{matrix code} with optimal repairing bandwidth {and the size of systematic node $F$,}
 an $ F$-division caching scheme for a $(K,M,N)$ caching system with $\frac{M}{N}=\frac{1}{r}$ can be obtained by $\mathcal{C}$.
\end{theorem}
\begin{proof}
From Lemma \ref{optimal bandwith}, there exist subspaces $S_{x,k}$ each of dimension $\frac{ F}{r}$, such that formula \eqref{property MSR} holds for any $k$, $k'\in[0,K)$ and $x\in [0,r)$. For each $k\in [0,K)$, we can define
\begin{eqnarray*}
S_k=S_{0,k}A_{0,k},\
S'_{k}=\left(\begin{array}{ccc}
S_{1,k} &        &\\
        &\ddots  & \\
        &        &S_{r-1,k}
\end{array}\right)\ \mbox{and}\
A_{k}=\left(\begin{array}{c}
A_{1,k}\\
\vdots \\
A_{r-1,k}
\end{array}\right).
\end{eqnarray*}
Clearly, for any $k$, $k'\in[k]$, we have
\begin{eqnarray*}
\mathrm{rank}\left(
  \begin{array}{c}
    S_k \\
    S'_{k}A_{k'}
 \end{array}
\right)=\mathrm{rank}\left(
  \begin{array}{c}
    S_{0,k}A_{0,k'} \\
    \vdots\\
    S_{r-1,k}A_{r-1,k'}
 \end{array}
\right)=\left\{\begin{array}{lll} F,&\mbox{if}\ \    k=k'
\\
\frac{ F}{r},&\mbox{otherwise}\end{array}\right.
\end{eqnarray*}
That completes the proof.
\end{proof}

From Theorem \ref{thm2}, to use MSR codes, we want the number of systematic nodes is as large as possible when we construct a linear coded caching scheme. Wang et al., in \cite{WTB} proposed a MSR code with optimal repair bandwidth such that the number of systematic nodes $K$ is the largest.
\begin{lemma}
\label{le-LongmDS}
There exists a $(n=q+K,K=(q+1)m)$ MSR with optimal repair bandwidth and the size of a node $F=q^m$ over $\mathbb{F}_p$ for any positive integer $m$ and $q$, where $p$ is a prime power and larger than $q$.
\end{lemma}
From Theorem \ref{thm2} and Lemma \ref{le-LongmDS}, the following result can be obtained.
\begin{corollary}
\label{cor-Longmds-caching}
For any positive integers $m$ and $q \geq 2$, there exists a $q^m$-division $((q+1)m,M,N)$ coded caching scheme with $\frac{M}{N}=\frac{1}{r}$ and $R=q-1$ over $\mathbb{F}_p$, where $p>q$ is a prime power.
\end{corollary}

Clearly for any positive integers $q\geq 2$ and $m$, when $z=1$ the scheme from Theorem \ref{th-main-R} and the scheme from Corollary \ref{cor-Longmds-caching} have all the same parameters except the computed field. However the scheme from Theorem \ref{th-main-R} is over $\mathbb{F}_2$ , while the scheme from Corollary \ref{cor-Longmds-caching} is over  $\mathbb{F}_p$, where $p>q$ is a prime power. In fact when constructing coding matrices $\mathbf{A}_{k}$ in Theorem \ref{thm2}, the MDS property is not necessary. It is interesting to know if
we can ignore the MDS property of the MSR to get better results.  So we have the following question.

\noindent{\bf Open problem:} How to construct linear coded aching schemes by modifying
 the constructions of MSR codes  for reducing the computed field or getting more classes of schemes.

Even though our work is completed independently to the work in \cite{WTB}, the used key sets in \eqref{Eqn_base vector1} and \eqref{Eqn_sum vector1} are the same. In \cite{WTB}, the authors first set that the encoding matrices $A_{0,k}$, $0\leq k<K$, equals identity matrix $I$. Then constructed each other encoding matrix, say $A_{i,k}$,  by choosing appropriate $q$ elements from set $\{\mathcal{E}_{u,v},\ \ \mathcal{Q}_{u}\ |\ 0\leq u<m, 0\le v<q\}$ as eigenvectors and choosing related $q$ non-zero elements from $\mathbb{F}_p$ as eigenvalues. For the details the interested reader could be referred to \cite{WTB}. However this method is not fit for further reducing the computed field when we construct a linear coded caching scheme.

\section{Conclusion}
\label{sec_conclusion}
In this paper, we consider efficient
constructions of coded caching schemes with good rate, lower subpacketization level and flexible number of users.
For this purpose,
  we first characterized a general coded caching scheme using linear algorithms, which generalized the previous
  main construction. Then we  showed that designing such a coded caching scheme depends
   on constructing three classes of matrices satisfying some rank conditions.
    We gave a concrete construction for several classes of new coded caching schemes over $\mathbb{F}_2$  by constructing these three classes of matrices. The rate of our new construction
    is the same as the scheme construct by Yan et al. in \cite{YCTC}, while the packet number is significantly smaller. This construction showed
    that the concept of  general linear caching scheme is very  promising for finding good code caching schemes.
    Finally,
    using the general definition of  linear caching scheme,
     we proved that the optimal  minimum storage regenerating codes can be used to construct coded caching schemes.

\vskip 1 cm


\section*{Appendix A: Proof of Theorem \ref{th-PDA-Linear}}
\begin{proof}
From the characterization of a linear coded caching scheme, it is sufficient to show that for any given a $(K,F,Z,S)$ PDA $\mathbf{P}=(p_{j,k})$, $0\leq j<F$, $0\leq k<K$, can be represented by the three classes of matrices $\mathbf{S}_k$, $\mathbf{A}_k$ and $\mathbf{S}'_k$, where $K$, $F$, $Z$ and $S$ are positive integers. Let $Z=\frac{FM}{N}$. For each $k=0$, $1$, $\ldots$, $K-1$, define a $Z\times F$ caching matrix
\begin{equation*}
\mathbf{S}_k=
\left(\begin{array}{c}
\psi(p_{0,k}){\bf e}_0\\
\psi(p_{1,k}){\bf e}_1\\
\vdots\\
\psi(p_{F-1,k}){\bf e}_{F-1}
\end{array}
\right).
\end{equation*}
where $\psi(p_{j,k})=1$ if $p_{j,k}=*$ otherwise $\psi(p_{j,k})=0$. For any request ${\bf d}$, servers broadcasts the following coded signals from Lines 8-10 of Algorithm \ref{alg:PDA}.
\begin{equation*}
\begin{array}{c}
\sum_{j\in[0,F),k\in[0,K),p_{j,k}=0}\mathbf{W}_{d_k,j},\\[0.3cm]
\sum_{j\in[0,F),k\in[0,K),p_{j,k}=1}\mathbf{W}_{d_k,j},\\
\vdots\\
\sum_{j\in[0,F),k\in[0,K),p_{j,k}=S-1}\mathbf{W}_{d_k,j}.
\end{array}
\end{equation*}
This is

\begin{eqnarray*}
\mathbf{X}_{{\bf d}}&=&\left(\begin{array}{c}
\sum_{j\in[0,F),k\in[0,K),p_{j,k}=0}\mathbf{W}_{d_k,j}\\
\sum_{j\in[0,F),k\in[0,K),p_{j,k}=1}\mathbf{W}_{d_k,j}\\
\vdots\\
\sum_{j\in[0,F),k\in[0,K),p_{j,k}=S-1}\mathbf{W}_{d_k,j}
\end{array}
\right)\\
&=&\sum\limits_{k=0}^{K-1}\left(\begin{array}{c}
\sum_{j\in[0,F),p_{j,k}=0}\mathbf{W}_{d_k,j}\\
\sum_{j\in[0,F),p_{j,k}=1}\mathbf{W}_{d_k,j}\\
\vdots\\
\sum_{j\in[0,F),p_{j,k}=S-1}\mathbf{W}_{d_k,j}
\end{array}
\right)=\sum\limits_{k=0}^{K-1}\left(\begin{array}{c}
\sum_{j\in[0,F)}\chi_0(p_{j,k})\mathbf{W}_{d_k,j}\\
\sum_{j\in[0,F)}\chi_1(p_{j,k})\mathbf{W}_{d_k,j}\\
\vdots\\
\sum_{j\in[0,F)}\chi_{S-1}(p_{j,k})\mathbf{W}_{d_k,j}
\end{array}
\right)\\
&=&
\sum\limits_{k=0}^{K-1}\left(\begin{array}{cccc}
\chi_0(p_{0,k})&\chi_0(p_{1,k})&\ldots&\chi_0(p_{F-1,k})\\
\chi_1(p_{0,k})&\chi_1(p_{1,k})&\ldots&\chi_1(p_{F-1,k})\\
\vdots\\
\chi_{S-1}(p_{0,k})&\chi_{S-1}(p_{1,k})&\ldots&\chi_{S-1}(p_{F-1,k})\\
\end{array}
\right)
\left(\begin{array}{c}
\mathbf{W}_{d_k,0}\\
\mathbf{W}_{d_k,1}\\
\vdots\\
\mathbf{W}_{d_k,F-1}
\end{array}
\right)\\
&=&\sum\limits_{k=0}^{K-1}\mathbf{A}_k
\left(\begin{array}{c}
\mathbf{W}_{d_k,0}\\
\mathbf{W}_{d_k,1}\\
\vdots\\
\mathbf{W}_{d_k,F-1}
\end{array}
\right)=\sum\limits_{k=0}^{K-1}\mathbf{A}_k\mathbf{W}_{d_k}
\end{eqnarray*}
where $\chi_s(p_{j,k})=1$ if $p_{j,k}=s$ otherwise $\psi(p_{j,k})=0$.
Then the coding matrices $\mathbf{A}_k$ can be obtained by the above equation.

Finally let us consider the decoding matrix $\mathbf{S}'_k$. From the definition of a PDA, we know that there are exactly $F-\frac{M}{N}F$ rows containing an integer $1$ in $\mathbf{A}_k$. Suppose these rows are  $i_1$, $i_2$, $\ldots$, $i_{F-\frac{M}{N}F}$. Define
\begin{eqnarray*}
\mathbf{S}'_k=
\left(
\begin{array}{c}
{\bf e}_{i_1}\\
{\bf e}_{i_2}\\
\ldots\\
{\bf e}_{i_{F-\frac{M}{N}F}}
\end{array}
\right)
\end{eqnarray*}
It is readily checked that these matrices forms the decoding matrices.
\end{proof}

\section*{Appendix B: Proof of Lemma \ref{lem1}}
\begin{proof}
We will verify the statements according to the cases $v_3<q$ and $v_3=q$.
\begin{itemize}
\item First let us show the statement (I). When $v_3<q$, by \eqref{Eqn_base vector1} and \eqref{eq-cons1-2} we have
\begin{eqnarray}
\label{eq-1}
\begin{split}
\mathcal{E}_{u_1,v_1}\mathbf{C}_{u_2,v_3,v_2}&=\mathcal{E}_{u_1,v_1}\left( \left(
\begin{matrix}
\phi_{u_2,v_3}(0){\bf e}_{\varphi_{u_2,v_2}(0)}\\
\phi_{u_2,v_3}(1){\bf e}_{\varphi_{u_2,v_2}(1)}\\
\ldots\\
\phi_{u_2,v_3}(q^m-1){\bf e}_{{\varphi_{u_2,v_2}(q^m-1)}}
\end{matrix}\right)+\left(
\begin{matrix}
\phi_{u_2,v_2}(0){\bf e}_0\\
\phi_{u_2,v_2}(1){\bf e}_1\\
\ldots\\
\phi_{u_2,v_2}(q^m-1){\bf e}_{q^m-1}
\end{matrix}\right)  \right)\\
&= \mathcal{E}_{u_1,v_1}\left(
\begin{matrix}
\phi_{u_2,v_3}(0){\bf e}_{\varphi_{u_2,v_2}(0)}\\
\phi_{u_2,v_3}(1){\bf e}_{\varphi_{u_2,v_2}(1)}\\
\ldots\\
\phi_{u_2,v_3}(q^m-1){\bf e}_{{\varphi_{u_2,v_2}(q^m-1)}}
\end{matrix}\right)+\mathcal{E}_{u_1,v_1}\left(
\begin{matrix}
\phi_{u_2,v_2}(0){\bf e}_0\\
\phi_{u_2,v_2}(1){\bf e}_1\\
\ldots\\
\phi_{u_2,v_2}(q^m-1){\bf e}_{q^m-1}
\end{matrix}\right)\\
&=\left\{{\bf e}_{\varphi_{u_2,v_2}(s)}\ |\ s\in
\mathcal{V}_{u_1,v_{1}}\bigcap\mathcal{V}_{u_2,v_{3}}\right\}+\left\{{\bf e}_s\ |\ s\in\mathcal{V}_{u_1,v_{1}}\bigcap\mathcal{V}_{u_2,v_{2}}\right\}\\
&=\mathcal{E}_{u_1,v_{1}}\bigcap\mathcal{E}_{u_2,v_{2}}\subseteq\mathcal{E}_{u_1,v_{1}}.
\end{split}
\end{eqnarray}
When $v_3=q$, by \eqref{Eqn_base vector1} and \eqref{eq-cons1-3} we have
\begin{eqnarray}
\label{eq-2}
\mathcal{E}_{u_1,v_1}\mathbf{C}_{u_2,q,v_2}&=&\mathcal{E}_{u_1,v_1} \left(
\begin{matrix}
\phi_{u_2,v_2}(0){\bf e}_0\\
\phi_{u_2,v_2}(1){\bf e}_1\\
\ldots\\
\phi_{u_2,v_2}(q^m-1){\bf e}_{q^m-1}
\end{matrix}\right)=\left\{{\bf e}_s\ |\ s\in\mathcal{V}_{u_1,v_{1}}\bigcap\mathcal{V}_{u_2,v_{2}}\right\}\subseteq\mathcal{E}_{u_1,v_{1}}.
\end{eqnarray}
So the statement (I) holds by \eqref{eq-1} and \eqref{eq-2}.
\item Now let us show the statement (II). When $v_3<q$. Similar to the proof of \eqref{eq-1} we have
\begin{eqnarray}
\label{eq-3}
\mathcal{E}_{u_1,v_1}\mathbf{C}_{u_1,v_2,v_2}=
\left\{{\bf e}_{\varphi_{u_1,v_2}(s)}\ |\ s\in
\mathcal{V}_{u_1,v_{1}}\bigcap\mathcal{V}_{u_1,v_{3}}\right\}+\left\{{\bf e}_s\ |\ s\in\mathcal{V}_{u_1,v_{1}}\bigcap\mathcal{V}_{u_1,v_{2}}\right\}
\end{eqnarray}
It is not difficult to check that
\begin{itemize}
\item if $v_1=v_3$ \eqref{eq-3} can be written as
\begin{eqnarray}
\label{eq-4}
\mathcal{E}_{u_1,v_1}\mathbf{C}_{u_1,v_2,v_2}=\mathcal{E}_{u_1,v_{2}}
+\left\{{\bf e}_s\ |\ s\in\mathcal{V}_{u_1,v_{1}}\bigcap\mathcal{V}_{u_1,v_{2}}\right\}=\mathcal{E}_{u_1,v_{2}},
\end{eqnarray}
\item and if $v_1\neq v_3$  \eqref{eq-3} can be written as
\begin{eqnarray}
\label{eq-5}
\mathcal{E}_{u_1,v_1}\mathbf{C}_{u_1,v_2,v_2}=\emptyset
+\left\{{\bf e}_s\ |\ s\in\mathcal{V}_{u_1,v_{1}}\bigcap\mathcal{V}_{u_2,v_{2}}\right\}\subseteq\mathcal{E}_{u_1,v_{1}}.
\end{eqnarray}
\end{itemize}
When $v_3=q$, similar to the proof of \eqref{eq-2}, we have
\begin{eqnarray}
\label{eq-6}
\mathcal{E}_{u_1,v_1}\mathbf{C}_{u_1,q,v_2} =\left\{{\bf e}_s\ |\ s\in\mathcal{V}_{u_1,v_{1}}\bigcap\mathcal{V}_{u_2,v_{2}}\right\}\subseteq\mathcal{E}_{u_1,v_{1}}.
\end{eqnarray}
So the statement (II) holds by \eqref{eq-4}, \eqref{eq-5} and \eqref{eq-6}.
\item When $v_3=q$, by \eqref{Eqn_sum vector1}, \eqref{eq-cons1-3} and \eqref{Eqn_base vector1} we have
\begin{eqnarray}
\label{eq-7}
\begin{split}
\mathcal{Q}_{u_1}\mathbf{C}_{u_2,q,v_2}&=\left(
\begin{matrix}
\sum\limits_{s_u=0}^{q-1}{\bf e}_{s}
\end{matrix}\right)_{s_{u'}\in[0,q),u'\in[0,m)\setminus\{u_1\}}\mathbf{C}_{u_2,q,v_2}
\\
&=\left(
\begin{matrix}
\sum\limits_{s_u=0}^{q-1}{\bf e}_{s}
\end{matrix}\right)_{s_{u'}\in[0,q),u'\in[0,m)\setminus\{u_1\}}\left(
\begin{matrix}
\phi_{u_2,v_2}(0){\bf e}_0\\
\phi_{u_2,v_2}(1){\bf e}_1\\
\ldots\\
\phi_{u_2,v_2}(q^m-1){\bf e}_{q^m-1}
\end{matrix}\right)\\
&=\left(
\begin{matrix}
\sum\limits_{s_u=0}^{q-1}{\bf e}_{s}\left(
\begin{matrix}
\phi_{u_2,v_2}(0){\bf e}_0\\
\phi_{u_2,v_2}(1){\bf e}_1\\
\ldots\\
\phi_{u_2,v_2}(q^m-1){\bf e}_{q^m-1}
\end{matrix}\right)
\end{matrix}\right)_{s_{u'}\in[0,q),u'\in[0,m)\setminus\{u_1\}}\\
&=\left(
\begin{matrix}
\sum\limits_{s_u=0}^{q-1}{\bf e}_{s}
\end{matrix}\right)_{s_{u'}\in[0,q),s_{u_2}=v_2,u'\in[0,m)\setminus\{u_1,u_2\}}
\end{split}
\end{eqnarray}
It is not difficult to check that
\begin{itemize}
\item if $u_1=u_2$ \eqref{eq-7} can be written as
\begin{eqnarray}
\label{eq-8}
\mathcal{Q}_{u_1}\mathbf{C}_{u_2,q,v_2}=\left(
\begin{matrix}
\sum\limits_{s_u=0}^{q-1}{\bf e}_{s}
\end{matrix}\right)_{s_{u'}\in[0,q),s_{u_1}=v_2,u'\in[0,m)\setminus\{u_1,u_2\}}=\mathcal{E}_{u_1,v_{2}},
\end{eqnarray}
\item and if $u_1\neq u_2$  \eqref{eq-7} can be written as
\begin{eqnarray}
\label{eq-9}
\begin{split}
\mathcal{Q}_{u_1}\mathbf{C}_{u_2,q,v_2}&=\left(
\begin{matrix}
\sum\limits_{s_u=0}^{q-1}{\bf e}_{s}
\end{matrix}\right)_{s_{u'}\in[0,q),s_{u_2}=v_2,u'\in[0,m)\setminus\{u_1,u_2\}}\\
&\subseteq \mathcal{Q}_{u_1}.
\end{split}
\end{eqnarray}
\end{itemize}
When $v_3<q$, by \eqref{Eqn_base vector1}, \eqref{Eqn_sum vector1}, \eqref{eq-cons1-2} and \eqref{eq-cons1-3} we have
 \begin{eqnarray}
\label{eq-10}
\begin{split}
&\mathcal{Q}_{u_1}\mathbf{C}_{u_2,v_3,v_2}\\
&=\left(
\begin{matrix}
\sum\limits_{s_u=0}^{q-1}{\bf e}_{s}
\end{matrix}\right)_{s_{u'}\in[0,q),u'\in[0,m)\setminus\{u_1\}}\mathbf{C}_{u_2,v_3,v_2}
\\
&=\left(
\begin{matrix}
\sum\limits_{s_u=0}^{q-1}{\bf e}_{s}
\end{matrix}\left( \left(
\begin{matrix}
\phi_{u_2,v_3}(0){\bf e}_{\varphi_{u_2,v_2}(0)}\\
\phi_{u_2,v_3}(1){\bf e}_{\varphi_{u_2,v_2}(1)}\\
\ldots\\
\phi_{u_2,v_3}(q^m-1){\bf e}_{{\varphi_{u_2,v_2}(q^m-1)}}
\end{matrix}\right)+\left(
\begin{matrix}
\phi_{u_2,v_2}(0){\bf e}_0\\
\phi_{u_2,v_2}(1){\bf e}_1\\
\ldots\\
\phi_{u_2,v_2}(q^m-1){\bf e}_{q^m-1}
\end{matrix}\right)  \right)\right)_{s_{u'}\in[0,q),u'\in[0,m)\setminus\{u_1\}}\\
&=\left({\bf e}_{\varphi_{u_2,v_2}(s')}+{\bf e}_{s}\right)_{s'_{u'}=s_{u'}\in[0,q),s'_{u_2}=v_3,s_{u_2}=v_2,u'\in[0,m)\setminus\{u_1,u_2\}}\\
&=0=\emptyset\subseteq\mathcal{Q}_{u_1}.
\end{split}
\end{eqnarray}
So the statement (III) holds by \eqref{eq-8}, \eqref{eq-9} and \eqref{eq-10}.
\end{itemize}
\end{proof}
\section*{Appendix C: Proof of Theorem \ref{th-main-R}}
\begin{proof}
From Construction \ref{con1}, for any fixed positive integers $u_1$, $u_2$, $v_1$, $v_2$, $\varepsilon_1$ and $\varepsilon_2$, let us show that the caching matrix $\mathbf{S}_{u_1,v_1,\varepsilon_1}$, coding matrix $\mathbf{A}_{u_2,v_2,\varepsilon_2}$ and decoding matrix $\mathbf{S}'_{u_1,v_1,\varepsilon_1}$ satisfy \eqref{Eqn rank condition} in Theorem \ref{th-main}. Let $\mathcal{G}_{v_1,\varepsilon_1}=\{v_{1,1},\ldots,v_{1,q-z}\}$ and $\mathcal{G}_{v_2,\varepsilon_2}=\{v_{2,1},\ldots,v_{2,q-z}\}$. Obviously we only need to consider \eqref{Eqn space condition} in the following case.
\begin{itemize}
\item When $(u_1,v_1,\varepsilon_1)=(u_2,v_2,\varepsilon_2)$, the following subcases can be obtained by \eqref{eq-caching1}.
\begin{itemize}
\item If $v_1<q$, from Lemma \ref{lem1}-(II) we have \begin{eqnarray*}
\mathbf{S}'_{u_1,v_1,\varepsilon_1}\mathbf{A}_{u_1,v_1,\varepsilon_1}
&=&
\left(\begin{matrix}\mathcal{E}_{u_1,v_1}&&\\ &\ddots& \\  && \mathcal{E}_{u_1,v_1}\end{matrix}\right)\left(\begin{matrix}\mathbf{C}_{u_1,v_1,v_{1,1}}\\ \mathbf{C}_{u_1,v_1,v_{1,2}}\\ \vdots \\ \mathbf{C}_{u_1,v_1,v_{1,q-z}} \end{matrix}\right)
=\left(\begin{matrix}\mathcal{E}_{u_1,v_1}\mathbf{C}_{u_1,v_1,v_{1,1}}\\ \mathcal{E}_{u_1,v_1}\mathbf{C}_{u_1,v_1,v_{1,1}}\\ \vdots \\ \mathcal{E}_{u_1,v_1}\mathbf{C}_{u_1,v_1,v_{1,q-z}} \end{matrix}\right) =\sum_{i=1}^{q-z}\mathcal{E}_{u_1,v_{1,i}}
\end{eqnarray*}
So we have
\begin{eqnarray*}
\mathbf{S}_{u_1,v_1,\varepsilon}+\mathbf{S}'_{u_1,v_1,\varepsilon}\mathbf{A}_{u_1,v_1,\varepsilon}
&=&\{\mathcal{E}_{u_1,v'}\ |\ v'\in [0,q)\setminus G_{v_1,\varepsilon_1}\} +\sum_{i=1}^{q-z} \mathcal{E}_{u_1,v_{2,i}}\\
&=&\sum_{v'\in [0,q)} \mathcal{E}_{u_1,v'}=(\mathbb{F}_2)^{q^m}
\end{eqnarray*}
Clearly the second equality of \eqref{Eqn space condition} holds.
\item If $v_1=q$, from Lemma \ref{lem1}-(III) we have
\begin{eqnarray*}
\mathbf{S}'_{u_1,q,\varepsilon}\mathbf{A}_{u_1,q,\varepsilon}
&=&
\left(\begin{matrix}\mathcal{Q}_{u_1}&&\\ &\ddots& \\  && \mathcal{Q}_{u_1}\end{matrix}\right)\left(\begin{matrix}\mathbf{C}_{u_1,q,v_{1,1}}\\ \mathbf{C}_{u_1,q,v_{1,2}}\\ \vdots \\ \mathbf{C}_{u_1,q,v_{1,q-z}} \end{matrix}\right)_{v_{1,i}\in \mathcal{G}_{q-1,\varepsilon}}=\left(\begin{matrix}\mathcal{Q}_{u_1}\mathbf{C}_{u_1,q,v_{1,1}}\\ \mathcal{Q}_{u_1}\mathbf{C}_{u_1,q,v_{1,2}}\\ \vdots \\ \mathcal{Q}_{u_1}\mathbf{C}_{u_1,q,v_{1,q-z}} \end{matrix}\right)\\
&=& \sum\limits_{i=1}^{q-z}\mathcal{E}_{u_1,v_{1,i}}
\end{eqnarray*}

So we have
\begin{eqnarray*}
\mathbf{S}_{u_1,q,\varepsilon_1}+\mathbf{S}'_{u_1,q,\varepsilon_1}\mathbf{A}_{u_1,q,\varepsilon_1}
&=&\mathcal{Q}_{u_1}\bigcup\{\mathcal{E}_{u_1,v'}\ |\ v'\in [0,q-1)\setminus G_{q-1,\varepsilon_1}\}+\sum\limits_{i=1}^{q-z}\mathcal{E}_{u_1,v_{1,i}}\\
&=&\mathcal{Q}_{u_1}\bigcup\{\mathcal{E}_{u_1,v'}\ |\ v'\in [0,q-1)\}\\
&=&\sum_{v'\in [0,q)} \mathcal{E}_{u_1,v'}=(\mathbb{F}_2)^{q^m}
\end{eqnarray*}
Clearly the second equality of \eqref{Eqn space condition} holds.
\end{itemize}
\item When $(u_1,v_1,\varepsilon_1)\neq(u_2,v_2,\varepsilon_2)$, we should consider the following subcases.
\begin{itemize}
\item If $u_1=u_2$, $v_1=v_2$ and $\varepsilon_1\neq \varepsilon_2$, let us consider the case $v_1<q$ first. From Lemma \ref{lem1}-(II) we have
\begin{eqnarray*}
\mathbf{S}'_{u_1,v_1,\varepsilon_1}\mathbf{A}_{u_1,v_2,\varepsilon_2}
&=&
\left(\begin{matrix}\mathcal{E}_{u_{1},v_1}&&\\ &\ddots& \\  && \mathcal{E}_{u_{1},v_1}\end{matrix}\right)\left(\begin{matrix}\mathbf{C}_{u_{1},v_2,v_{2,1}}\\ \mathbf{C}_{u_{1},v_2,v_{2,1}}\\ \vdots \\ \mathbf{C}_{u_{1},v_2,v_{2,q-z}} \end{matrix}\right)\\
&=&\left(\begin{matrix}\mathcal{E}_{u_{1},v_1}\mathbf{C}_{u_{1},v_2,v_{2,1}}\\ \mathcal{E}_{u_{1},v_1}\mathbf{C}_{u_{1},v_2,v_{2,2}}\\ \vdots \\ \mathcal{E}_{u_{1},v_1}\mathbf{C}_{u_{1},v_2,v_{2,q-z}} \end{matrix}\right)\\
&=&\sum\limits_{i=1}^{q-z}\mathcal{E}_{u_{1},v_{2,i}}.
\end{eqnarray*}

Since $0\leq\varepsilon_1\neq\varepsilon_2<h$ and $v_1=v_2$, $\mathcal{G}_{v_1,\varepsilon_1}\bigcap\mathcal{G}_{v_2,\varepsilon_2}=\emptyset$ always holds. This implies that $\mathcal{G}_{v_2,\varepsilon_2}\subseteq [0,q)\setminus\mathcal{G}_{v_1,\varepsilon_1}$. So we have
$$\sum\limits_{i=1}^{q-z}\mathcal{E}_{u_1,v_{2,i}}\subseteq\{\mathcal{E}_{u,v'}\ |\ v'\in [0,q)\setminus G_{v_1,\varepsilon_1}\}=\mathbf{S}_{u_1,v_1,\varepsilon_1}.$$
Clearly the first equality of \eqref{Eqn space condition} holds. Let us consider the case $v_1=q$ now. From Lemma \ref{lem1}-(III) we have
\begin{eqnarray*}
\mathbf{S}'_{u_1,q,\varepsilon_1}\mathbf{A}_{u_1,q,\varepsilon_2}
&=&
\left(\begin{matrix}\mathcal{Q}_{u_1}&&\\ &\ddots& \\  && \mathcal{Q}_{u_1}\end{matrix}\right)\left(\begin{matrix}\mathbf{C}_{u_1,q,v_{2,1}}\\ \mathbf{C}_{u_1,q,v_{2,2}}\\ \vdots \\ \mathbf{C}_{u_1,q,v_{2,q-z}} \end{matrix}\right) =\left(\begin{matrix}\mathcal{Q}_{u_1}\mathbf{C}_{u_1,q,v_{2,1}}\\ \mathcal{Q}_{u_1}\mathbf{C}_{u_1,q,v_{2,2}}\\ \vdots \\ \mathcal{Q}_{u_1}\mathbf{C}_{u_1,q,v_{2,q-z}} \end{matrix}\right)\\
&=&\sum\limits_{i=1}^{q-z}\mathcal{E}_{u,v_{2,i}}
\end{eqnarray*}
Since $0\leq \varepsilon_1\neq\varepsilon<h$ and $v_1=v_2$, $\mathcal{G}_{v_1,\varepsilon_1}\bigcap\mathcal{G}_{v_2,\varepsilon_2}=\emptyset$ always holds. This implies that $\mathcal{G}_{v_2,\varepsilon_2}\subseteq [0,q-1)\setminus\mathcal{G}_{v_1,\varepsilon_1}$. So we have $$\sum\limits_{i=1}^{q-z}\mathcal{E}_{u_1,v_{2,i}}\subseteq\{\mathcal{E}_{u_1,v'}\ |\ v'\in [0,q-1)\setminus G_{v_1,\varepsilon_1}\}\subseteq\mathbf{S}_{u_1,v_1,\varepsilon_1}.$$
Clearly the first equality of \eqref{Eqn space condition} holds.
\item If $u_1=u_2$ and $v_1\neq v_2$, we also first consider the case $v_1<q$. From Lemma \ref{lem1}-(II) we have
\begin{eqnarray}
\label{eq-11}
\begin{split}
\mathbf{S}'_{u_1,v_1,\varepsilon_1}\mathbf{A}_{u_2,v_2,\varepsilon_2}
&=\left(\begin{matrix}\mathcal{E}_{u_1,v_1}&&\\ &\ddots& \\  && \mathcal{E}_{u_1,v_1}\end{matrix}\right)\left(\begin{matrix}\mathbf{C}_{u_2,v_2,v_{2,1}}\\ \mathbf{C}_{u_2,v_2,v_{2,2}}\\ \vdots \\ \mathbf{C}_{u_2,v_2,v_{2,q-z}} \end{matrix}\right)
=\left(\begin{matrix}\mathcal{E}_{u_1,v_1}\mathbf{C}_{u_2,v_2,v_{2,1}}\\ \mathcal{E}_{u_1,v_1}\mathbf{C}_{u_2,v_2,v_{2,2}}\\ \vdots \\ \mathcal{E}_{u_1,v_1}\mathbf{C}_{u_2,v_2,v_{2,q-z}} \end{matrix}\right) \\
&\subseteq\mathcal{E}_{u_1,v_1}
\end{split}
\end{eqnarray}
Clearly the first equality of \eqref{Eqn space condition} holds. Now let us consider the case $v_1=q$. Clearly $v_2<q$. From Lemma \ref{lem1}-(III) we have
\begin{eqnarray}
\label{eq-12}
\begin{split}
\mathbf{S}'_{u_1,v_1,\varepsilon_1}\mathbf{A}_{u_2,v_2,\varepsilon_2}
&=\left(\begin{matrix}\mathcal{Q}_{u_1}&&\\ &\ddots& \\  && \mathcal{Q}_{u_1}\end{matrix}\right)\left(\begin{matrix}\mathbf{C}_{u_2,v_2,v_{2,1}}\\ \mathbf{C}_{u_2,v_2,v_{2,2}}\\ \vdots \\ \mathbf{C}_{u_2,v_2,v_{2,q-z}} \end{matrix}\right)
=\left(\begin{matrix}\mathcal{Q}_{u_1}\mathbf{C}_{u_2,v_2,v_{2,1}}\\ \mathcal{Q}_{u_1}\mathbf{C}_{u_2,v_2,v_{2,2}}\\ \vdots \\ \mathcal{Q}_{u_1}\mathbf{C}_{u_2,v_2,v_{2,q-z}} \end{matrix}\right) \\
&\subseteq\mathcal{Q}_{u_1}
\end{split}
\end{eqnarray}
\item If $u_1\neq u_2$, for the cases $v_1<q$ and $v_1=q$, \eqref{eq-11} and \eqref{eq-12} can be obtained by Lemma \ref{lem1}-(I) and (III) respectively. So the first equality of \eqref{Eqn space condition} holds.
\end{itemize}
\end{itemize}
\end{proof}


\begin{thebibliography}{99}
\bibitem{AG}
M. M. Amiri and D. G\"{u}nd\"{u}z, Fundamental limits of caching: Improved delivery rate-cache capacity trade-off, {\em IEEE Transactions on Communications}, vol. 65, no. 2, pp. 806-815, 2016.
\bibitem{ART}
A. Sengupta, R. Tandon, and T. Clancy, Fundamental limits of caching with secure delivery, {\em IEEE Trans. Inf. Forensics Security}, vol. 10, no. 2, pp. 355-370, 2015.






\bibitem{CJYT}
M. Cheng, J. Jiang, Q. Yan and  X. Tang, Coded caching schemes generated by partitions for flexible memory sizes, arXiv:1708.06650 [cs,IT], Aug. 2017.


\bibitem{CJYW}
M. Cheng, J. Jiang, Y. Yao, and Q. Wang, A novel Recursive construction for coded caching
schemes, arXiv:1712.09090v1[cs.IT], Dec. 2017.


\bibitem{CYTJ}
M. Cheng, Q. Yan, X. Tang, and J. Jiang ,Coded caching schemes with low rate and subpacketizations, arXiv:1703.01548v2 [cs.IT], Apr 2017.





\bibitem{DGWWR}
A. G. Dimakis, P. B. Godfrey, Y. Wu, M. J. Wainwright, and
K. Ramchandran, Network coding for distributed storage systems, IEEE Trans. Inf. Theory, vol. 56, no. 9, pp. 4539-551, 2010.

\bibitem{GR}
H. Ghasemi and A. Ramamoorthy, Further results on lower bounds for coded caching, in Proc. {\em IEEE International Symposium on Information Theory}, Barcelona, Spai, July, 2016, pp. 2319-2323.




\bibitem{JCM}
M. Ji, G. Caire and A. F. Molisch, Fundamental Limits of Caching in Wireless D2D Networks, {\em IEEE Transactions on Information Theory}, vol. 62, no. 2, pp.849-869, 2016.

\bibitem{JWTL}
M. Ji, M. Wong, A. M. Tulino, J. Llorca, On the fundamental limits of caching in combination networks, in Proc. {\em 2015 IEEE 16th International Workshop on Signal Processing Advances in Wireless Communications (SPAWC)}, Stockholm, Sweden, July, 2015, pp. 695-699.

\bibitem{JCLC}
S. Jin, Y. Cui , H. Liu, and G. Caire, Uncoded placement optimization for coded delivery, {\em IEEE WiOpt}, Shanghai, China, May. 2018.


\bibitem{K}
P. Krishnan, Prasad. Coded Caching via Line Graphs of Bipartite Graphs, arXiv preprint arXiv:1805.08955 [cs.IT], May 2018.

\bibitem{KNMD}
N. Karamchandani, U. Niesen, M. A. Maddah-Ali, and S. Diggavi, Hierarchical coded caching, in Proc. {\em IEEE International Symposium on Information Theory}, Honolulu, HI, Jun. 2014, pp. 2142-2146.


\bibitem{MN}
M. A. Maddah-Ali and U. Niesen, Fundamental limits of caching, {\em IEEE Transactions on Information Theory},
vol. 60, no. 5, pp. 2856-2867, 2014.

\bibitem{MN1}
M. A. Maddah-Ali, Urs Niesen, Decentralized coded caching attains order-optimal memory-rate tradeoff, {\em EEE/ACM Transactions on Networking}, vol. 23, no. 4, pp.1029-1040, 2015.



\bibitem{RAN}
R. Pedarsani, M. A. Maddah-Ali, and U. Niesen, Online coded caching, in Proc. {\em  IEEE International Conference on Communications}, Sydney, Australia, Jun. 2014, pp. 1878-1883.

\bibitem{SZG}
C. Shangguan, Y. Zhang, G. Ge, Centralized coded caching schemes: A hypergraph theoretical approach,  {\em IEEE Transactions on Information Theory,} vol. 64, no. 8, pp. 5755-5766, 2018.

\bibitem{STC}
A. Sengupta, R. Tandon, and T. C. Clancy, Improved approximation of
storage-rate tradeoff for caching via new outer bounds, in Proc. {\em IEEE International Symposium on Information Theory}, 2015, pp. 1691-1695.

\bibitem{SJTLD}
K. Shanmugam, M. Ji, A. M. Tulino, J. Llorca, and A. G. Dimakis, Finite-length analysis of caching-aided coded multicasting, {\em IEEE Transactions on Information Theory}, vol.~62, no.~10, pp. 5524--5537, 2016.

\bibitem{STD}
K. Shanmugam, A. M. Tulino, and A. G. Dimakis, Coded caching with linear subpacketization is possible using Ruzsa-Szem\'{e}redi graphs, in Proc. {\em IEEE International Symposium on Information Theory}, Aachen, Germany, Jun. 2017, pp. 1237-1241.

\bibitem{SDLT}
K. Shanmugam, A. G. Dimakis, J. Llorca and A. M. Tulino, A unified Ruzsa-Szemer\'{e}di framework for finite-length coded caching, 2017 51st {\em Asilomar Conference on Signals, Systems, and Computers}, Pacific Grove, CA, 2017, pp. 631-635.


\bibitem{TWB} T. Tamo, Z. Wang and J. Bruck, Access versus bandwidth in codes for storage, {\em IEEE Transactions on Information Theory,} vol. 60, no. 4, pp. 2028-2037, 2014.

\bibitem{TC}
C. Tian and J. Chen, Caching and delivery via interference elimination,  {\em IEEE Transactions on Information Theory}, vol. 64, no. 3, pp. 1548-1560, 2018.

\bibitem{TR}
L. Tang, A. Ramamoorthy, Coded caching schemes with reduced subpacketization from linear block codes, IEEE Trans. Inform. Theory, vol. 64, no. 4, pp. 3099-3120, 2018.

\bibitem{WTB} Z. Wang, T. Tamo and J. Bruck, Long MDS codes for optimal repair bandwidth, Tech. Rep. Available at  \textit{http :
//paradise.caltech.edu/etr.html}.



\bibitem{WLG}
C. Y. Wang, S. H. Lim, M. Gastpar, A new converse bound for coded caching, in Proc. {\em IEEE Information Theory Workshop}, Robinson College, Oct. 2016.

\bibitem{WTP}
K. Wan, D. Tuninetti and P. Piantanida, On the optimality of uncoded cache placement, in Proc. {\em IEEE Information Theory Workshop}, Cambridge, UK, Sept. 2016.


\bibitem{YCTC}
Q. Yan, M. Cheng, X. Tang and Q. Chen, On the placement delivery array design in centralized coded caching scheme, {\em IEEE Transactions on Information Theory}, vol.~63, no.~9, pp. 5821-5833, 2017.

\bibitem{YMA}
Q. Yu, M. A. Maddah-Ali, A. S. Avestimehr,
The exact rate-memory tradeoff for caching with uncoded prefetching, {\em IEEE Transactions on Information Theory}, vol. 64, no. 2, pp. 1281-1296, 2018.


\bibitem{YTCC}
Q. Yan, X. Tang, Q. Chen, M. Cheng, Placement delivery array design through strong edge coloring of bipartite graphs, {\em IEEE Communications Letters,} vol. 22, no. 2, pp. 236-239, Feb. 2018


\end{thebibliography}
\end{document}